\documentclass[envcountsame]{llncs}
\usepackage{enumerate,multirow}
\usepackage{amsmath}
\usepackage{amsfonts,amssymb,mathtools,xspace}
\usepackage{tikz}
\usetikzlibrary{decorations.pathreplacing}
\usetikzlibrary{decorations.markings}
\usetikzlibrary{calc}
\usetikzlibrary{shapes.geometric}
\usetikzlibrary{positioning}
\usepackage{etoolbox}
\usepackage{tabularx}
\usepackage{framed} % or, "mdframed"

\usepackage{changepage}% http://ctan.org/pkg/changepage
\spnewtheorem{axiom}{Axiom}{\bfseries}{\itshape}

\newenvironment{pic}[1][]
{\begin{aligned}\begin{tikzpicture}[#1]}
{\end{tikzpicture}\end{aligned}}
\newcommand{\edges}[1][]%
{%\end{scope}\end{pgfonlayer}\begin{pgfonlayer}{foreground}\begin{scope}[#1]
}

\makeatletter
\def\calign@preamble{%
   &\hfil\strut@
    \setboxz@h{\@lign$\m@th\displaystyle{##}$}%
    \ifmeasuring@\savefieldlength@\fi
    \set@field
    \hfil
    \tabskip\alignsep@
}
\let\cmeasure@\measure@
\patchcmd\cmeasure@{\divide\@tempcntb\tw@}{}{}{}
\patchcmd\cmeasure@{\divide\@tempcntb\tw@}{}{}{}
\patchcmd\cmeasure@{\ifodd\maxfields@
  \global\advance\maxfields@\@ne
  \fi}{}{}{}
\newenvironment{calign}
{%
  \let\align@preamble\calign@preamble
  \let\measure@\cmeasure@
  \align
}
{%
  \endalign
}
\makeatother


\newcommand\tinymatrix[1]
{\left( \hspace{-2pt} \renewcommand\thickspace{\kern2pt} \scriptstyle\begin{smallmatrix} #1 \end{smallmatrix} \hspace{-2pt} \right)}

\newcommand\ignore[1]{}

\pgfdeclarelayer{foreground}
\pgfdeclarelayer{background}
\pgfsetlayers{main,foreground,background}
\tikzset{smallbox/.style={draw, fill=white, minimum height=0.45cm, minimum width=0.45cm, inner sep=-100pt}}

\makeatletter
\pgfkeys{%
  /tikz/on layer/.code={
    \pgfonlayer{#1}\begingroup
    \aftergroup\endpgfonlayer
    \aftergroup\endgroup
  },
  /tikz/node on layer/.code={
    \gdef\node@@on@layer{%
      \setbox\tikz@tempbox=\hbox\bgroup\pgfonlayer{#1}\unhbox\tikz@tempbox\endpgfonlayer\egroup}
    \aftergroup\node@on@layer
  },
  /tikz/end node on layer/.code={
    \endpgfonlayer\endgroup\endgroup
  }
}
\def\node@on@layer{\aftergroup\node@@on@layer}
\makeatother\def\thickness{0.7pt}
\tikzstyle{oldmorphism}=[minimum width=30pt, minimum height=16pt, draw, font=\small, inner sep=0pt, fill=white, line width=\thickness]
\tikzstyle{cross}=[preaction={draw=white, -, line width=10pt}]
\tikzstyle{braid}=[double=black, line width=3*\thickness, double distance=\thickness, white]
\tikzstyle{string}=[line width=\thickness]
\tikzstyle{scalar}=[circle, inner sep=0pt, minimum width=15pt, draw, line width=\thickness]
\tikzstyle{dot}=[circle, draw=black, fill=black!25, inner sep=.4ex, line width=\thickness, node on layer=foreground]
\tikzstyle{blackdot}=[circle, draw=black, fill=black!100, inner sep=.4ex, line width=\thickness, node on layer=foreground]
\tikzstyle{graydot}=[circle, draw=black, fill=gray!40!white, inner sep=.4ex, line width=\thickness, node on layer=foreground]
\tikzstyle{whitedot}=[circle, draw=black, fill=white, inner sep=.4ex, line width=\thickness, node on layer=foreground]
\tikzstyle{mixedmorphism}=[morphism, minimum width=30pt, minimum height=16pt, draw, font=\small, inner sep=0pt, fill=white, line width=\thickness,rounded corners=1ex]
\tikzstyle{thick}=[line width=\thickness]
\tikzstyle{tiny}=[font=\tiny]

\tikzset{arrow/.style={decoration={
    markings,
    mark=at position #1 with \arrow{thickarrow}},
    postaction=decorate}
}
\tikzset{reverse arrow/.style={decoration={
    markings,
    mark=at position #1 with \arrow{reversethickarrow}},
    postaction=decorate}
}

\newif\ifblack\pgfkeys{/tikz/black/.is if=black}
\newif\ifwedge\pgfkeys{/tikz/wedge/.is if=wedge}
\newif\ifvflip\pgfkeys{/tikz/vflip/.is if=vflip}
\newif\ifhflip\pgfkeys{/tikz/hflip/.is if=hflip}
\newif\ifhvflip\pgfkeys{/tikz/hvflip/.is if=hvflip}
\newif\ifconnectnw\pgfkeys{/tikz/connect nw/.is if=connectnw}
\newif\ifconnectne\pgfkeys{/tikz/connect ne/.is if=connectne}
\newif\ifconnectsw\pgfkeys{/tikz/connect sw/.is if=connectsw}
\newif\ifconnectse\pgfkeys{/tikz/connect se/.is if=connectse}
\newif\ifconnectn\pgfkeys{/tikz/connect n/.is if=connectn}
\newif\ifconnects\pgfkeys{/tikz/connect s/.is if=connects}
\newif\ifconnectnwf\pgfkeys{/tikz/connect nw >/.is if=connectnwf}
\newif\ifconnectnef\pgfkeys{/tikz/connect ne >/.is if=connectnef}
\newif\ifconnectswf\pgfkeys{/tikz/connect sw >/.is if=connectswf}
\newif\ifconnectsef\pgfkeys{/tikz/connect se >/.is if=connectsef}
\newif\ifconnectnf\pgfkeys{/tikz/connect n >/.is if=connectnf}
\newif\ifconnectsf\pgfkeys{/tikz/connect s >/.is if=connectsf}
\newif\ifconnectnwr\pgfkeys{/tikz/connect nw </.is if=connectnwr}
\newif\ifconnectner\pgfkeys{/tikz/connect ne </.is if=connectner}
\newif\ifconnectswr\pgfkeys{/tikz/connect sw </.is if=connectswr}
\newif\ifconnectser\pgfkeys{/tikz/connect se </.is if=connectser}
\newif\ifconnectnr\pgfkeys{/tikz/connect n </.is if=connectnr}
\newif\ifconnectsr\pgfkeys{/tikz/connect s </.is if=connectsr}
\tikzset{keylengthnw/.initial=\connectheight}
\tikzset{keylengthn/.initial =\connectheight}
\tikzset{keylengthne/.initial=\connectheight}
\tikzset{keylengthsw/.initial=\connectheight}
\tikzset{keylengths/.initial =\connectheight}
\tikzset{keylengthse/.initial=\connectheight}
\tikzset{connect nw length/.style={connect nw=true, keylengthnw={#1}}}
\tikzset{connect n length/.style ={connect n =true, keylengthn ={#1}}}
\tikzset{connect ne length/.style={connect ne=true, keylengthne={#1}}}
\tikzset{connect sw length/.style={connect sw=true, keylengthsw={#1}}}
\tikzset{connect s length/.style ={connect s =true, keylengths ={#1}}}
\tikzset{connect se length/.style={connect se=true, keylengthse={#1}}}
\tikzset{connect nw < length/.style={connect nw <=true, keylengthnw={#1}}}
\tikzset{connect n < length/.style ={connect n <=true,  keylengthn ={#1}}}
\tikzset{connect ne < length/.style={connect ne <=true, keylengthne={#1}}}
\tikzset{connect sw < length/.style={connect sw <=true, keylengthnw={#1}}}
\tikzset{connect s < length/.style ={connect s <=true,  keylengths ={#1}}}
\tikzset{connect se < length/.style={connect se <=true, keylengthse={#1}}}
\tikzset{connect nw > length/.style={connect nw >=true, keylengthnw={#1}}}
\tikzset{connect n > length/.style ={connect n >=true,  keylengthn ={#1}}}
\tikzset{connect ne > length/.style={connect ne >=true, keylengthne={#1}}}
\tikzset{connect sw > length/.style={connect sw >=true, keylengthsw={#1}}}
\tikzset{connect s > length/.style ={connect s >=true,  keylengths ={#1}}}
\tikzset{connect se > length/.style={connect se >=true, keylengthse={#1}}}

\newlength\morphismheight
\newlength\minimummorphismwidth
\newlength\stateheight
\newlength\minimumstatewidth
\newlength\connectheight
\tikzset{width/.initial=\minimummorphismwidth}

\makeatletter
\pgfarrowsdeclare{thickarrow}{thickarrow}
{
  \pgfutil@tempdima=-0.84pt%
  \advance\pgfutil@tempdima by-1.3\pgflinewidth%
  \pgfutil@tempdimb=-1.7pt%
  \advance\pgfutil@tempdimb by.625\pgflinewidth%
  \pgfarrowsleftextend{+\pgfutil@tempdima}
  \pgfarrowsrightextend{+\pgfutil@tempdimb}
}
{
  \pgfmathparse{\pgfgetarrowoptions{thickarrow}}%
  \pgfsetlinewidth{1.25 pt}
  \pgfutil@tempdima=0.28pt%
  \advance\pgfutil@tempdima by.3\pgflinewidth%
  \pgfsetlinewidth{0.8\pgflinewidth}
  \pgfsetdash{}{+0pt}
  \pgfsetroundcap
  \pgfsetroundjoin
  \pgfpathmoveto{\pgfqpoint{-3\pgfutil@tempdima}{4\pgfutil@tempdima}}
  \pgfpathcurveto
  {\pgfqpoint{-2.75\pgfutil@tempdima}{2.5\pgfutil@tempdima}}
  {\pgfqpoint{0pt}{0.25\pgfutil@tempdima}}
  {\pgfqpoint{0.75\pgfutil@tempdima}{0pt}}
  \pgfpathcurveto
  {\pgfqpoint{0pt}{-0.25\pgfutil@tempdima}}
  {\pgfqpoint{-2.75\pgfutil@tempdima}{-2.5\pgfutil@tempdima}}
  {\pgfqpoint{-3\pgfutil@tempdima}{-4\pgfutil@tempdima}}
  \pgfusepathqstroke
}
\pgfarrowsdeclare{reversethickarrow}{reversethickarrow}
{
  \pgfutil@tempdima=-0.84pt%
  \advance\pgfutil@tempdima by-1.3\pgflinewidth%
  \pgfutil@tempdimb=0.2pt%
  \advance\pgfutil@tempdimb by.625\pgflinewidth%
  \pgfarrowsleftextend{+\pgfutil@tempdima}
  \pgfarrowsrightextend{+\pgfutil@tempdimb}
}
{
  \pgftransformxscale{-1}
  \pgfmathparse{\pgfgetarrowoptions{thickarrow}}%
  \ifpgfmathunitsdeclared%
    \pgfmathparse{\pgfmathresult pt}%
  \else%
    \pgfmathparse{\pgfmathresult*\pgflinewidth}%
  \fi%
  \let\thickness=\pgfmathresult
  \pgfsetlinewidth{1.25 pt}
  \pgfutil@tempdima=0.28pt%
  \advance\pgfutil@tempdima by.3\pgflinewidth%
  \pgfsetlinewidth{0.8\pgflinewidth}
  \pgfsetdash{}{+0pt}
  \pgfsetroundcap
  \pgfsetroundjoin
  \pgfpathmoveto{\pgfqpoint{-3\pgfutil@tempdima}{4\pgfutil@tempdima}}
  \pgfpathcurveto
  {\pgfqpoint{-2.75\pgfutil@tempdima}{2.5\pgfutil@tempdima}}
  {\pgfqpoint{0pt}{0.25\pgfutil@tempdima}}
  {\pgfqpoint{0.75\pgfutil@tempdima}{0pt}}
  \pgfpathcurveto
  {\pgfqpoint{0pt}{-0.25\pgfutil@tempdima}}
  {\pgfqpoint{-2.75\pgfutil@tempdima}{-2.5\pgfutil@tempdima}}
  {\pgfqpoint{-3\pgfutil@tempdima}{-4\pgfutil@tempdima}}
  \pgfusepathqstroke
}
\makeatother

\makeatletter
\pgfdeclareshape{ground}
{
    \savedanchor\centerpoint
    {
        \pgf@x=0pt
        \pgf@y=0pt
    }
    \anchor{center}{\centerpoint}
    \anchorborder{\centerpoint}
    \anchor{north}
    {
        \pgf@x=0pt
        \pgf@y=0.5*0.33*\stateheight
    }
    \anchor{south}
    {
        \pgf@x=0pt
        \pgf@y=0pt
    }
    \saveddimen\overallwidth
    {
        \pgfkeysgetvalue{/pgf/minimum width}{\minwidth}
        \pgf@x=\minimumstatewidth
        \ifdim\pgf@x<\minwidth
            \pgf@x=\minwidth
        \fi
    }
    \backgroundpath
    {
        \begin{pgfonlayer}{foreground}
        \pgfsetstrokecolor{black}
        \pgfsetlinewidth{1.25pt}
        \ifhflip
            \pgftransformyscale{-1}
        \fi
        \pgftransformscale{0.5}
        \pgfpathmoveto{\pgfpoint{-0.5*\overallwidth}{0}}
        \pgfpathlineto{\pgfpoint{0.5*\overallwidth}{0}}
        \pgfpathmoveto{\pgfpoint{-0.33*\overallwidth}{0.33*\stateheight}}
        \pgfpathlineto{\pgfpoint{0.33*\overallwidth}{0.33*\stateheight}}
        \pgfpathmoveto{\pgfpoint{-0.16*\overallwidth}{0.66*\stateheight}}
        \pgfpathlineto{\pgfpoint{0.16*\overallwidth}{0.66*\stateheight}}
        \pgfpathmoveto{\pgfpoint{-0.02*\overallwidth}{\stateheight}}
        \pgfpathlineto{\pgfpoint{0.02*\overallwidth}{\stateheight}}
        \pgfusepath{stroke}
        \end{pgfonlayer}
    }
}
\tikzset{forward arrow style/.style={every to/.style, decoration={
    markings,
    mark=at position 0.5 with \arrow{thickarrow}},
    postaction=decorate}}
\tikzset{reverse arrow style/.style={every to/.style, decoration={
    markings,
    mark=at position 0.5 with \arrow{reversethickarrow}},
    postaction=decorate}}
\pgfdeclareshape{morphism}
{
    \savedanchor\centerpoint
    {
        \pgf@x=0pt
        \pgf@y=0pt
    }
    \anchor{center}{\centerpoint}
    \anchorborder{\centerpoint}
    \saveddimen\savedlengthnw
    {
        \pgfkeysgetvalue{/tikz/keylengthnw}{\len}
        \pgf@x=\len
    }
    \saveddimen\savedlengthn
    {
        \pgfkeysgetvalue{/tikz/keylengthn}{\len}
        \pgf@x=\len
    }
    \saveddimen\savedlengthne
    {
        \pgfkeysgetvalue{/tikz/keylengthne}{\len}
        \pgf@x=\len
    }
    \saveddimen\savedlengthsw
    {
        \pgfkeysgetvalue{/tikz/keylengthsw}{\len}
        \pgf@x=\len
    }
    \saveddimen\savedlengths
    {
        \pgfkeysgetvalue{/tikz/keylengths}{\len}
        \pgf@x=\len
    }
    \saveddimen\savedlengthse
    {
        \pgfkeysgetvalue{/tikz/keylengthse}{\len}
        \pgf@x=\len
    }
    \saveddimen\overallwidth
    {
        \pgfkeysgetvalue{/tikz/width}{\minwidth}
        \pgf@x=\wd\pgfnodeparttextbox
        \ifdim\pgf@x<\minwidth
            \pgf@x=\minwidth
        \fi
    }
    \savedanchor{\upperrightcorner}
    {
        \pgf@y=.5\ht\pgfnodeparttextbox
        \advance\pgf@y by -.5\dp\pgfnodeparttextbox
        \pgf@x=.5\wd\pgfnodeparttextbox
    }
    \anchor{north}
    {
        \pgf@x=0pt
        \pgf@y=0.5\morphismheight
    }
    \anchor{north east}
    {
        \pgf@x=\overallwidth
        \multiply \pgf@x by 2
        \divide \pgf@x by 5
        \pgf@y=0.5\morphismheight
    }
    \anchor{east}
    {
        \pgf@x=\overallwidth
        \divide \pgf@x by 2
        \advance \pgf@x by 5pt
        \pgf@y=0pt
    }
    \anchor{west}
    {
        \pgf@x=-\overallwidth
        \divide \pgf@x by 2
        \advance \pgf@x by -5pt
        \pgf@y=0pt
    }
    \anchor{north west}
    {
        \pgf@x=-\overallwidth
        \multiply \pgf@x by 2
        \divide \pgf@x by 5
        \pgf@y=0.5\morphismheight
    }
    \anchor{connect nw}
    {
        \pgf@x=-\overallwidth
        \multiply \pgf@x by 2
        \divide \pgf@x by 5
        \pgf@y=0.5\morphismheight
        \advance\pgf@y by \savedlengthnw
    }
    \anchor{connect ne}
    {
        \pgf@x=\overallwidth
        \multiply \pgf@x by 2
        \divide \pgf@x by 5
        \pgf@y=0.5\morphismheight
        \advance\pgf@y by \savedlengthne
    }
    \anchor{connect sw}
    {
        \pgf@x=-\overallwidth
        \multiply \pgf@x by 2
        \divide \pgf@x by 5
        \pgf@y=-0.5\morphismheight
        \advance\pgf@y by -\savedlengthsw
    }
    \anchor{connect se}
    {
        \pgf@x=\overallwidth
        \multiply \pgf@x by 2
        \divide \pgf@x by 5
        \pgf@y=-0.5\morphismheight
        \advance\pgf@y by -\savedlengthse
    }
    \anchor{connect n}
    {
        \pgf@x=0pt
        \pgf@y=0.5\morphismheight
        \advance\pgf@y by \savedlengthn
    }
    \anchor{connect s}
    {
        \pgf@x=0pt
        \pgf@y=-0.5\morphismheight
        \advance\pgf@y by -\savedlengths
    }
    \anchor{south east}
    {
        \pgf@x=\overallwidth
        \multiply \pgf@x by 2
        \divide \pgf@x by 5
        \pgf@y=-0.5\morphismheight
    }
    \anchor{south west}
    {
        \pgf@x=-\overallwidth
        \multiply \pgf@x by 2
        \divide \pgf@x by 5
        \pgf@y=-0.5\morphismheight
    }
    \anchor{south}
    {
        \pgf@x=0pt
        \pgf@y=-0.5\morphismheight
    }
    \anchor{text}
    {
        \upperrightcorner
        \pgf@x=-\pgf@x
        \pgf@y=-\pgf@y
    }
    \backgroundpath
    {
        \pgfsetstrokecolor{black}
        \pgfsetlinewidth{\thickness}
        \begin{scope}
                \ifhflip
                    \pgftransformyscale{-1}
                \fi
                \ifvflip
                    \pgftransformxscale{-1}
                \fi
                \ifhvflip
                    \pgftransformxscale{-1}
                    \pgftransformyscale{-1}
                \fi
                \pgfpathmoveto{\pgfpoint
                    {-0.5*\overallwidth-5pt}
                    {0.5*\morphismheight}}
                \pgfpathlineto{\pgfpoint
                    {0.5*\overallwidth+5pt}
                    {0.5*\morphismheight}}
                \ifwedge
                    \pgfpathlineto{\pgfpoint
                        {0.5*\overallwidth + 15pt}
                        {-0.5*\morphismheight}}
                \else
                    \pgfpathlineto{\pgfpoint
                        {0.5*\overallwidth + 5pt}
                        {-0.5*\morphismheight}}
                \fi
                \pgfpathlineto{\pgfpoint
                    {-0.5*\overallwidth-5pt}
                    {-0.5*\morphismheight}}
                \pgfpathclose
                \pgfusepath{stroke}
        \end{scope}
        \ifconnectnw
            \pgfpathmoveto{\pgfpoint
                {-0.4*\overallwidth}
                {0.5*\morphismheight}}
            \pgfpathlineto{\pgfpoint
                {-0.4*\overallwidth}
                {0.5*\morphismheight+\savedlengthnw}}
            \pgfusepath{stroke}
        \fi
        \ifconnectne
            \pgfpathmoveto{\pgfpoint
                {0.4*\overallwidth}
                {0.5*\morphismheight}}
            \pgfpathlineto{\pgfpoint
                {0.4*\overallwidth}
                {0.5*\morphismheight+\savedlengthne}}
            \pgfusepath{stroke}
        \fi
        \ifconnectsw
            \pgfpathmoveto{\pgfpoint
                {-0.4*\overallwidth}
                {-0.5*\morphismheight}}
            \pgfpathlineto{\pgfpoint
                {-0.4*\overallwidth}
                {-0.5*\morphismheight-\savedlengthsw}}
            \pgfusepath{stroke}
        \fi
        \ifconnectse
            \pgfpathmoveto{\pgfpoint
                {0.4*\overallwidth}
                {-0.5*\morphismheight}}
            \pgfpathlineto{\pgfpoint
                {0.4*\overallwidth}
                {-0.5*\morphismheight-\savedlengthse}}
            \pgfusepath{stroke}
        \fi
        \ifconnectn
            \pgfpathmoveto{\pgfpoint
                {0pt}
                {0.5*\morphismheight}}
            \pgfpathlineto{\pgfpoint
                {0pt}
                {0.5*\morphismheight+\savedlengthn}}
            \pgfusepath{stroke}
        \fi
        \ifconnects
            \pgfpathmoveto{\pgfpoint
                {0pt}
                {-0.5*\morphismheight}}
            \pgfpathlineto{\pgfpoint
                {0pt}
                {-0.5*\morphismheight-\savedlengths}}
            \pgfusepath{stroke}
        \fi
        \ifconnectnwf
            \draw [forward arrow style] (-0.4*\overallwidth,0.5*\morphismheight)
                to (-0.4*\overallwidth,0.5*\morphismheight+\savedlengthnw);
        \fi
        \ifconnectnef
            \draw [forward arrow style] (0.4*\overallwidth,0.5*\morphismheight)
                to (0.4*\overallwidth,0.5*\morphismheight+\savedlengthne);
        \fi
        \ifconnectswf
            \draw [forward arrow style] (-0.4*\overallwidth,-0.5*\morphismheight-\savedlengthsw)
                to (-0.4*\overallwidth,-0.5*\morphismheight);
        \fi
        \ifconnectsef
            \draw [forward arrow style] (0.4*\overallwidth,-0.5*\morphismheight-\savedlengthse)
                to (0.4*\overallwidth,-0.5*\morphismheight);
        \fi
        \ifconnectnf
            \draw [forward arrow style] (0,0.5*\morphismheight)
                to (0,0.5*\morphismheight+\savedlengthn);
        \fi
        \ifconnectsf
            \draw [forward arrow style] (0,-0.5*\morphismheight-\savedlengths)
                to (0,-0.5*\morphismheight);
        \fi
        \ifconnectnwr
            \draw [reverse arrow style] (-0.4*\overallwidth,0.5*\morphismheight)
                to (-0.4*\overallwidth,0.5*\morphismheight+\savedlengthnw);
        \fi
        \ifconnectner
            \draw [reverse arrow style] (0.4*\overallwidth,0.5*\morphismheight)
                to (0.4*\overallwidth,0.5*\morphismheight+\savedlengthne);
        \fi
        \ifconnectswr
            \draw [reverse arrow style] (-0.4*\overallwidth,-0.5*\morphismheight-\savedlengthsw)
                to (-0.4*\overallwidth,-0.5*\morphismheight);
        \fi
        \ifconnectser
            \draw [reverse arrow style] (0.4*\overallwidth,-0.5*\morphismheight-\savedlengthse)
                to (0.4*\overallwidth,-0.5*\morphismheight);
        \fi
        \ifconnectnr
            \draw [reverse arrow style] (0,0.5*\morphismheight)
                to (0,0.5*\morphismheight+\savedlengthn);
        \fi
        \ifconnectsr
            \draw [reverse arrow style] (0,-0.5*\morphismheight-\savedlengths)
                to (0,-0.5*\morphismheight);
        \fi
    }
}
\pgfdeclareshape{swish right}
{
    \savedanchor\centerpoint
    {
        \pgf@x=0pt
        \pgf@y=0pt
    }
    \anchor{center}{\centerpoint}
    \anchorborder{\centerpoint}
    \anchor{north}
    {
        \pgf@x=\minimummorphismwidth
        \divide\pgf@x by 5
        \pgf@y=\morphismheight
        \divide\pgf@y by 2
        \advance\pgf@y by \connectheight
    }
    \anchor{south}
    {
        \pgf@x=-\minimummorphismwidth
        \divide\pgf@x by 5
        \pgf@y=-\morphismheight
        \divide\pgf@y by 2
        \advance\pgf@y by -\connectheight
    }
    \backgroundpath
    {
        \pgfsetstrokecolor{black}
        \pgfsetlinewidth{\thickness}
        \pgfpathmoveto{\pgfpoint
            {-0.2*\minimummorphismwidth}
            {-0.5*\morphismheight-\connectheight}}
        \pgfpathcurveto
            {\pgfpoint{-0.2*\minimummorphismwidth}{0pt}}
            {\pgfpoint{0.2*\minimummorphismwidth}{0pt}}
            {\pgfpoint
                {0.2*\minimummorphismwidth}
                {0.5*\morphismheight+\connectheight}}
        \pgfusepath{stroke}
    }
}
\pgfdeclareshape{swish left}
{
    \savedanchor\centerpoint
    {
        \pgf@x=0pt
        \pgf@y=0pt
    }
    \anchor{center}{\centerpoint}
    \anchorborder{\centerpoint}
    \anchor{north}
    {
        \pgf@x=-\minimummorphismwidth
        \divide\pgf@x by 5
        \pgf@y=\morphismheight
        \divide\pgf@y by 2
        \advance\pgf@y by \connectheight
    }
    \anchor{south}
    {
        \pgf@x=\minimummorphismwidth
        \divide\pgf@x by 5
        \pgf@y=-\morphismheight
        \divide\pgf@y by 2
        \advance\pgf@y by -\connectheight
    }
    \backgroundpath
    {
        \pgfsetstrokecolor{black}
        \pgfsetlinewidth{\thickness}
        \pgfpathmoveto{\pgfpoint
            {0.2*\minimummorphismwidth}
            {-0.5*\morphismheight-\connectheight}}
        \pgfpathcurveto
            {\pgfpoint{0.2*\minimummorphismwidth}{0pt}}
            {\pgfpoint{-0.2*\minimummorphismwidth}{0pt}}
            {\pgfpoint
                {-0.2*\minimummorphismwidth}
                {0.5*\morphismheight+\connectheight}}
        \pgfusepath{stroke}
    }
}
\pgfdeclareshape{state}
{
    \savedanchor\centerpoint
    {
        \pgf@x=0pt
        \pgf@y=0pt
    }
    \anchor{center}{\centerpoint}
    \anchorborder{\centerpoint}
    \saveddimen\overallwidth
    {
        \pgf@x=3\wd\pgfnodeparttextbox
        \ifdim\pgf@x<\minimumstatewidth
            \pgf@x=\minimumstatewidth
        \fi
    }
    \savedanchor{\upperrightcorner}
    {
        \pgf@x=.5\wd\pgfnodeparttextbox
        \pgf@y=.5\ht\pgfnodeparttextbox
        \advance\pgf@y by -.5\dp\pgfnodeparttextbox
    }
    \anchor{A}
    {
        \pgf@x=-\overallwidth
        \divide\pgf@x by 4
        \pgf@y=0pt
    }
    \anchor{B}
    {
        \pgf@x=\overallwidth
        \divide\pgf@x by 4
        \pgf@y=0pt
    }
    \anchor{text}
    {
        \upperrightcorner
        \pgf@x=-\pgf@x
        \ifhflip
            \pgf@y=-\pgf@y
            \advance\pgf@y by 0.4\stateheight
        \else
            \pgf@y=-\pgf@y
            \advance\pgf@y by -0.4\stateheight
        \fi
    }
    \backgroundpath
    {
        \pgfsetstrokecolor{black}
        \pgfsetlinewidth{\thickness}
        \pgfpathmoveto{\pgfpoint{-0.5*\overallwidth}{0}}
        \pgfpathlineto{\pgfpoint{0.5*\overallwidth}{0}}
        \ifhflip
            \pgfpathlineto{\pgfpoint{0}{\stateheight}}
        \else
            \pgfpathlineto{\pgfpoint{0}{-\stateheight}}
        \fi
        \pgfpathclose
        \ifblack
            \pgfsetfillcolor{black!50}
            \pgfusepath{fill,stroke}
        \else
            \pgfusepath{stroke}
        \fi
    }
}
\makeatother

\newcommand{\tinycomult}[1][dot]{
\smash{\raisebox{-1.8pt}{\hspace{-5pt}\ensuremath{\begin{pic}[scale=0.34,string]
    \node (0) at (0,0) {};
    \node[#1, inner sep=1.5pt] (1) at (0,0.55) {};
    \node (2) at (-0.5,1) {};
    \node (3) at (0.5,1) {};
    \draw (0.center) to (1.center);
    \draw (1.center) to [out=left, in=down, out looseness=1.5] (2.center);
    \draw (1.center) to [out=right, in=down, out looseness=1.5] (3.center);
\end{pic}
}\hspace{-3pt}}}}
\newcommand{\tinycounit}[1][dot]{
\smash{\raisebox{-1.8pt}{\ensuremath{\hspace{-3pt}\begin{pic}[scale=0.34,string]
        \node (0) at (0,0) {};
        \node (1) at (0,1) {};
        \node[#1, inner sep=1.5pt] (d) at (0,0.55) {};
        \draw (0.center) to (d.center);
    \end{pic}
    \hspace{-1pt}}}}}
\newcommand{\tinymult}[1][dot]{
\smash{{\hspace{-5pt}\ensuremath{\begin{aligned}\begin{tikzpicture}[scale=0.36,thick,yscale=-1]
    \node (0) at (0,0) {};
    \node (2) at (-0.5,1) {};
    \node (3) at (0.5,1) {};
    \draw (0.center) to (0,0.55);
    \draw (0,0.55) to [out=left, in=down, out looseness=1.5] (2.center);
    \draw (0,0.55) to [out=right, in=down, out looseness=1.5] (3.center);
    \node[#1, inner sep=1.5pt] (1) at (0,0.55) {};
\end{tikzpicture}\end{aligned}
}\hspace{-3pt}}}}
\newcommand{\tinyunit}[1][dot]{
\smash{{\ensuremath{\hspace{-3pt}\begin{aligned}\begin{tikzpicture}[scale=0.4,thick,yscale=-1]
        \node (0) at (0,0) {};
        \node (1) at (0,1) {};
        \draw (0.center) to (0,0.55);
        \node[#1, inner sep=1.5pt] (d) at (0,0.55) {};
    \end{tikzpicture}\end{aligned}
    \hspace{-1pt}}}}}

\newcommand{\tinyhandle}[1][dot]{\raisebox{-2pt}{\ensuremath{\hspace{-3pt}\begin{pic}[scale=0.4,string]
        \node (0) at (0,0) {};
        \node[dot, inner sep=1.0pt] (1) at (0,0.3) {};
        \node[dot, inner sep=1.0pt] (2) at (0,0.7) {};
        \node (3) at (0,1) {};
        \draw (0.center) to (1.center);
        \draw (2.center) to (3.center);
        \draw[in=180, out=180, looseness=2] (1.center) to (2.center);
        \draw[in=0, out=0, looseness=2] (1.center) to (2.center);
\end{pic}\hspace{-1pt}}}}

\newcommand\cat[1]{\ensuremath{\mathbf{#1}}}

\newcommand\ket[1]{\ensuremath{| #1 \rangle}}
\newcommand\bra[1]{\ensuremath{\langle #1 |}}

\usepackage{color}

\newcommand{\idm}[1]{\mbox{id}_{#1}}

\def\swangle{-145}
\def\seangle{-35}
\def\nwangle{145}
\def\neangle{35}

\ignore{
\usepackage{mathtools}
\tikzstyle{dot}=[inner sep=0.7mm,minimum width=0pt,minimum
height=0pt,fill=black,draw=black,shape=circle]
\tikzstyle{black dot}=[dot,fill=black]
\tikzstyle{white dot}=[dot,fill=white]
\tikzstyle{gray dot}=[dot,fill=gray!40!white]
}

\newcommand\whitecomonoid[1]{\ensuremath{(#1,\tinycomult[whitedot],\tinycounit[whitedot])}}
\newcommand\blackcomonoid[1]{\ensuremath{(#1,\tinycomult[blackdot],\tinycounit[blackdot])}}
\newcommand\graycomonoid[1]{\ensuremath{(#1,\tinycomult[graydot],\tinycounit[graydot])}}

\usepackage{graphicx}

\begin{document}

%+Title
\title{Models of Quantum Algorithms in Sets and Relations}
%\titlerunning{Quantum Algorithms in Sets and Relations}
\author{William Zeng}
%\authorrunning{W. J. Zeng}
\institute{Department of Computer Science, University of Oxford, Oxford, UK \\
\email{william.zeng@cs.ox.ac.uk}}
\date{\today}
\maketitle
%-Title

%+Abstract
\begin{abstract}
        We construct abstract models of blackbox quantum algorithms using a model of quantum computation in sets and relations, a setting that is usually considered for nondeterministic classical computation.  This alternative model of quantum computation (QCRel), though unphysical, nevertheless faithfully models its computational structure.  Our main results are models of the Deutsch-Jozsa, single-shot Grovers, and GroupHomID algorithms in QCRel. These results provide new tools to analyze the semantics of quantum computation and improve our understanding of the relationship between computational speedups and the structure of physical theories. They also exemplify a method of extending physical/computational intuition into new mathematical settings.
\end{abstract}
%-Abstract

\section{Introduction}

Despite almost two decades of research, we still seek new and useful quantum algorithms.  This is of interest in cases where the meaning of useful ranges from ``able to generate experimental evidence against the extended Church-Turing thesis" to ``commercially viable". Better languages, frameworks, and techniques for analyzing the structure of quantum algorithms will aid in these attempts.  One such  programme initiated by Abramsky, Coecke, et al. de-emphasizes the role of Hilbert spaces and linear maps and instead focuses on topological flows of information within quantum-like systems
\cite{coecke-abramsky-cqm,InteractQOb,coecke-pavlovic-2006}. This approach captures all the familiar structure of quantum computation from teleportation  to quantum secret-sharing and locates the particular quantum setting of Hilbert spaces as an instance of more general abstract process theories \cite{qcs-notes,coecke2011categories}.  Recent work has developed the presentation and verification of quantum algorithms such as the Deutsch-Jozsa and Grover algorithms, and the quantum Fourier transform \cite{strongCompFT} in terms of these abstract process theories, finding new generalizations and algorithms \cite{vicary-tqa,zeng-unitary}.

Having grasped the abstract structure at play in the protocols and algorithms of quantum computation, we can conceive of modelling quantum computation in settings other than Hilbert spaces and linear maps.  There are two main thrusts that make this investigation, the subject of this paper, interesting.  The first is to further analyze the structure of quantum computation, advancing our understanding of the relationship between computational speedups and the structure of physical theories. We use the QCRel model defined here to analyze some example quantum algorithms as non-deterministic classical algorithms while preserving their query-complexity (and, in fact, all their abstract structure). The second thrust regards the insights that become available by extending physical/computational intuition into new areas of mathematics. While other toy models of a relational flavor for quantum mechanics have been proposed \cite{ellermanModelQM,discreteQT,modalQT,spekk}, and some even discuss protocols \cite{QCFF_James}, these works have not developed the structures necessary to model quantum algorithms.

The next section of this paper constructs our chosen model of quantum information.  This is the setting of sets and relations, rather than Hilbert spaces and linear maps, and it is introduced by rephrasing the axioms of quantum mechanics. Section 3 introduces a graphical notation for analyzing processes in this setting. Sections 4-9 present the novel contributions of this paper: relational models of unitary oracles, the Deutsch-Jozsa algorithm, the single-shot Grover's algorithm, and the group homomorphism identification algorithm.

\subsubsection{Acknowledgements}
The author would especially like to acknowledge the useful discussions and encouragement from Bob Coecke, Chris Heunen, Jamie Vicary and anonymous reviewers as well as funding support from The Rhodes Trust and AFOSR grant FA9550-14-1-0079.

\section{The Model of Quantum Computation in Relations}

We begin by defining the key components of quantum computation in this new setting, e.g. systems, states, bases, etc.  The following definitions are motivated by examples from \cite{coecke-abramsky-cqm,coecke2011categories,evans2009classifying} that are summarized in \cite{qcs-notes,cqm-notes}, whose general theorems prove useful. To avoid distracting repetition of notation, we use generic terminology to refer to the relational setting within this paper.  For example \emph{system} is intended to mean \emph{relational system}, i.e. a set.  When we wish to refer to the quantum setting we explicitly denote this e.g. \emph{quantum system} refers to a finite dimensional Hilbert space.

\begin{axiom}
A \emph{system} is a set $H$ with \emph{states} $|\psi\rangle$ given by subsets $\psi\subseteq H$.
\end{axiom}

\noindent Each state in our notation is a boolean column vector written as a labelled ket, to follow the convention in quantum mechanics where states are complex valued column vectors as in the following example.

\begin{example}
\label{ex:columnvect}
Consider a three element system $\{0,1,2\}$, the relation $R=\{(0,0),(0,2),(1,1)\}$ and the state $\ket{\psi}=\{0\}$. In terms of boolean matrices and vectors the composition $R\circ\ket{\psi}$ is written as:
\begin{equation}
\left(\begin{array}{ccc}
1 & 0 & 0 \\
0 & 1 & 0 \\
1 & 0 & 0 \\
\end{array}\right)
\left(\begin{array}{c}
1 \\
0 \\
0 \\
\end{array}\right)
=
\left(\begin{array}{c}
1 \\
0 \\
1 \\
\end{array}\right).
\end{equation}
\end{example}
The state $|\psi\vee\phi\rangle$ has elements in the union of sets $\psi$ and $\phi$. We often use $\ket{\psi}$ to mean the relation $\{\bullet\}\to H$ that relates the singleton set to all the elements in $\psi$.

\begin{axiom}
A \emph{composite system} $H$ of $n$ subsystems is given by the Cartesian product so that $H = H_1\times...\times H_n$. \emph{Composite states} are any subset of $H$.
\end{axiom}

\begin{definition}
For a relation $R:A\to B$ from set $A$ to $B$, the \emph{converse relation} is denoted $R^{-1}:B\to A$ where for $x\in A$ and $y\in B$, $xRy$ if and only if $yR^{-1}x$.
\end{definition}

\noindent The converse replaces the $\dagger$-adjoint in quantum mechanics. This leads to:

\begin{definition}
A relation $R:H_1\to H_2$ is \emph{unitary} if and only if $R\circ R^{-1} = \mbox{id}_{H_1}$ and $R^{-1}\circ R = \mbox{id}_{H_2}$, where $Q\circ R$ means $Q$ after $R$.
\end{definition}

\noindent This is the relational analog to the usual unitarity of linear maps in quantum mechanics and has an obvious interpretation:

\begin{corollary}
\label{cor:bijections}
Relations are unitary if and only if they are bijections.
\end{corollary}

\begin{axiom}
\emph{Evolution} of systems is given by unitary relations.
\end{axiom}

\noindent This means that states of system $A$ can evolve to states of system $B$ if and only if there is a bijection between them. Note that this implies that there do not exist physical evolutions between systems of different cardinality. %This is analogous to the quantum setting where there do not exist unitary %maps between Hilbert spaces of differing dimensions.

\begin{definition}
For a state $\ket{\psi}:\{\bullet\}\to H$, denote its relational converse as $\bra{\psi}:H\to\{\bullet\}$ called its \emph{effect}.
\end{definition}

\noindent A state preparation followed by an effect amounts to an experiment with a post-selected outcome. Effects are maps to $\{\bullet\}$ that return whether the outcome state $\ket{\psi}$ is possible.
 We give an example to illustrate:
\begin{example}
The preparation of the state $\ket{\phi}$ followed by a post-selected measurement of the effect $\bra{\psi}$ is given by the relation
\begin{align*}
\langle\psi|\phi\rangle:=\bra{\psi}\circ\ket{\phi}:\{\bullet\}\to H\to \{\bullet\}
\end{align*}
This is either the identity relation that we interpret to mean a measurement outcome of $\bra{\psi}$ is \emph{possible}, or it is the empty relation that we interpret to mean the measurement outcome $\bra{\psi}$ is \emph{impossible}. It is clear that the outcome $\bra{\psi}$ is possible if there exists some element of $H$ in both $\psi$ and $\phi$. Otherwise it is impossible.
In this sense our relational quantum computation is a deterministic model of quantum computation.
\end{example}

This interpretation allows us to define a generalized version of the Born rule\footnote{In quantum theory, the Born rule gives the probability of measuring the outcome state $\ket{\psi}$ following preparation in state $\ket{\phi}$ as $|\langle\psi|\phi\rangle|^2$ where $\langle\psi|\phi\rangle:\mathbb{C}\to\mathbb{C}$ is the inner product of the two state vectors \cite{nielsen2010quantum}.} to describe measurement in our model.
\begin{axiom}[Generalized Born Rule]
\label{ax:born}
The possibility of measuring the state $\ket{\psi}$, having prepared state $\ket{\phi}$, is given by the image of:
\begin{align}
\langle\psi|\phi\rangle:\{\bullet\}\to\{\bullet\}
\end{align}
\end{axiom}

In the relational model, bases are characterized as particular generalizations of groups known as \textit{groupoids} \cite{cqm-notes,pavlovic-2009}.  Groupoids can be viewed as groups where multiplication is relaxed to be a partial function.

\begin{definition}
\label{def:basis}
For a system $H$, a \emph{basis} $Z$ is a direct sum (disjoint union) of abelian groups $Z = G_0\oplus G_1\oplus...$ where $|Z| = |H|$.
Multiplication with respect to this list of groups will be written as $\bullet_Z$ and is defined in the following way. For elements $x,y\in Z$ such that $x\in G_i$ and $y\in G_j$ we have the partial function:
\begin{align}
\label{eq:groupoid_mult}
x\bullet_Zy = \begin{cases}
x +_{G_i} y & i=j \\
\mbox{undefined} & \mbox{otherwise}\\
\end{cases}
\end{align}
\noindent This makes $Z$ an \emph{abelian groupoid} with groupoid multiplication $\bullet_Z$.
\end{definition}

We will sometimes take a categorical perspective on groupoids. A groupoid $Z=\bigoplus^NG_i$ made up $N$ groups is a category whose set of objects is isomorphic to the set of groups $\{G_i\}$ and whose morphisms are elements of $Z$, e.g. $x\in Z$ such that $x\in G_1$ is a morphism $x:G_1\to G_1$.

At first guess, one might be motivated by the intuition that a basis for a system breaks it up into parts, and so a basis would be a partition of $H$.  This is not a bad start, however, bases have additional structure: namely that we can copy, delete and combine them at will.  This idea is used to motivate Definition \ref{def:basis} by abstracting bases to special dagger-commutative Frobenius algebras (Definition \ref{def:classicalstruct}) that we call classical structures \cite{coecke-pavlovic-vicary-2008}.

Classical structures' properties, allowing the copying, deleting, and combining that accompany classical (as opposed to quantum) information, give them this name. The definition of a special dagger-commutative Frobenius algebra in our model is given in Section \ref{section:graphical}, and we can interpret it through pair of lemmas corresponding to the traditional model and the relational model of quantum computation.
\begin{lemma}[\cite{coecke-pavlovic-vicary-2008}]
\label{lem:sdfa-hilb}
The classical structures in the category of finite dimensional Hilbert spaces and linear maps are exactly the orthonormal bases.
\end{lemma}

\begin{lemma}[\cite{evans2009classifying,pavlovic-2009}]
\label{lem:sdfa-rel}
The classical structures in the category of sets and relations are exactly the abelian groupoids.\footnote{In \cite{heunen-relFrob} this connection is extended to the non-abelian case where it is shown that all relative Frobenius algebras are groupoids.}
\end{lemma}

\subsection{Complementarity}
Complementary bases are important features of quantum theory. In the general setting, complementary bases are understood as mutually unbiased bases in a certain sense \cite{InteractQOb}.  In relations there is a more direct characterization:\footnote{Theorem \ref{thm:compl} holds as long as we consider bases to be the same if their lists of groups are isomorphic.}
\begin{theorem}[\cite{evans2009classifying}]
\label{thm:compl}
Two bases $Z$ and $X$ are complementary if and only if they are of the following form. Basis $Z = \bigoplus^{|H|}G$ and basis $X = \bigoplus^{|G|}H$ given by copies of abelian groups $G$ and $H$ respectively.
\end{theorem}

This theorem follows from the requirement that the classical states of one basis must be isomorphic to the unbiased states of its complement. We will return to this idea in the Section \ref{sec:FT} when we address the quantum Fourier transform. Classical and unbiased states of bases in the relational model are specified in the following definitions that instantiate abstract definitions in \cite{InteractQOb}. An example on the six element system is illustrated with Figure \ref{complEx}.

\begin{figure}[tb]
\begin{center}
\includegraphics[height=10em,natwidth=1091,natheight=468,scale=1]{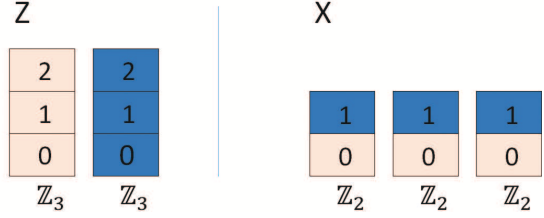}
\end{center}
\vspace{-14pt}
\caption{An example of two complementary bases on the system of six elements. Here $Z=\mathbb{Z}_3\oplus\mathbb{Z}_3$ and $X = \mathbb{Z}_2\oplus\mathbb{Z}_2\oplus \mathbb{Z}_2$.  The two classical states of $Z$ are each three element subsets and are colored in pink and blue. The unbiased states of $X$ to which they correspond are colored to match.
}
\label{complEx}
\end{figure}

\begin{definition}[\cite{evans2009classifying}]
The \emph{classical states} of a basis $Z = \bigoplus^{N}G$ are the subsets corresponding to the groups $G_0, G_1,...$ where we forget the group structure. They will often be denoted $\ket{G_i}$.
\end{definition}

\begin{definition}[\cite{evans2009classifying}]
The \emph{unbiased states} for a basis $Z = \bigoplus^{N}G$ are subsets $U$ such that for a fixed $g\in G$, $\ket{U} = \bigoplus^{N}\{g\}$.
Thus there is exactly one element in each unbiased $U$ from each component $G_i$ of $Z$.
\end{definition}

\begin{example}
Take $Z = \mathbb{Z}_2\oplus\mathbb{Z}_2=\{0_a,1_a,0_b,1_b\}$. The classical states of $Z$ are $\ket{G_a}=\ket{0_a\vee1_a}$ and $\ket{G_b}=\ket{0_b\vee1_b}$.  The unbiased states of $Z$ are $\ket{U_0}=\ket{0_a\vee0_b}$ and $\ket{U_1}=\ket{1_a\vee1_b}$.
\end{example}

It is easy to check that bases as specified by Theorem \ref{thm:compl} have the property that each classical state $\ket{G_i}$ of the basis $Z$ corresponds to one unbiased state of $X$ and vice versa.
This allows us to call these bases mutually unbiased, i.e. complementary \cite{evans2009classifying}.
%\subsection{Phases}

%Phases are also defined in this relational setting.  In Hilbert space quantum %mechanics a quantum phase for an n-dimensional system is given by the vector
%\begin{align*}
%\left(\begin{array}{c}
%e^{i\phi_1} \\
%\vdots \\
%e^{i\phi_n}
%\end{array}
%\right).
%\end{align*}
%These quantum phases form an abelian group and can be applied as phase gates.
%Their relational counterparts are described by the following lemma from %\cite{cqm-notes}:
%\begin{lemma}
%For a basis $Z=\bigoplus_i^NG_i$, the \emph{phase group} $B(Z)$ is given %by $\prod_i^NG_i$.
%\end{lemma}

%\begin{example}
%Consider the basis $\mathbb{Z}_2\oplus\mathbb{Z}_2$ for the four element %system $\{00,01,10,11\}$.  Let $|\psi\rangle$ be the state $|00\vee10\rangle$. %Application of the phase $11$ results in
%\[ 11|00\vee10\rangle = |11\vee01\rangle . \]
%\end{example}

%We are also able to interpret GHZ states and density matrices in sets and %relations.

%\begin{definition}
%For a basis $Z$, a \emph{GHZ} state is given by
%\[ GHZ_Z \; := \{\;(a,b,c)\;|\;\ \forall \;a,b,c \in Z,\;a\bullet_Zb\bullet_Zc %= \mbox{id}_{G_i}\mbox{ for some } i\;\}.  \]
%\end{definition}

%\begin{definition}
%For a state $|\psi\rangle$, the \emph{density matrix} $|\psi\rangle\langle\psi|$ %is given by the relation $xRy$ s.t. $x,y\in \psi$.
%\end{definition}

\subsection{The Model QCRel}
\label{sec:model}
\begin{definition}
Axioms 1-4, and subsequent definitions, specify the abstract process theory for \emph{quantum computation in relations: QCRel}.
\end{definition}

\begin{theorem}
QCRel is a model of quantum computation with sets and unitary relations.
\end{theorem}
\begin{proof}
This is true by construction.  The axioms on the preceding section can be interpreted as structures in any dagger compact category. In particular, {\bf FHilb}, the category of finite dimensional Hilbert spaces and linear maps,  is a dagger compact category in which interpretation of those axioms results in the usual Hilbert space quantum mechanics \cite{coecke-abramsky-cqm}.  {\bf Rel}, the category of sets and relations, is also dagger compact.  It is the interpretation of the abstract axioms for quantum computation in {\bf Rel}, rather than {\bf FHilb}, that produces QCRel as a model. References that covers the dagger compact abstraction and some of its interpretation in different categories are \cite{coecke2011categories,cqm-notes}.
\end{proof}

It is worth noting that QCRel can be simply viewed as a local hidden variable theory. We consider the set $H$ to be the set of ontic states such that for $\phi\subseteq H$ the state $\ket{\phi}$ is non-deterministically in any of the ontic states in the subset $\phi$.  From this perspective, QCRel provides a non-deterministic local hidden variable model for computational aspects of quantum mechanics \cite{abramsky2012operational}. This means that protocols exist for entanglement, teleportation, and, as we show in this paper, some familiar blackbox algorithms.
\vspace{-10pt}
\section{Graphical Notation}
\label{section:graphical}
In this section, we introduce a simple graphical notation that is commonly used in the literature for abstract process theories \cite{qcs-notes,coecke2011categories}. See~\cite{selinger} for a survey of these diagrams.  This notation will ease the inclusion of higher level proofs in our particular setting. In the context of this paper, this graphical notation acts as a more formal circuit-like model to present protocols and algorithms. For systems $A$ and $B$, we represent the relation $f:A\to B$ as:
%\vspace{-15pt}
\begin{align*}
    \begin{tikzpicture}[yscale=0.6]
                \node (0) at (0, 1) {};
                \node (1) at (0.3, 1) {$B$};
                \node [style=morphism] (2) at (0, -0) {$f$};
                \node (3) at (0, -1) {};
                \node (4) at (0.3, -1) {$A$};
                \draw (2.north) to (0.center);
                \draw (3.center) to (0,-0.45);
    \end{tikzpicture}
\end{align*}
%\vspace{-20pt}

\noindent ``reading" the diagram from bottom to top.  We represent individual systems (the identity morphism on them), sequential composition, states, and composite systems with the following diagrams, where relations are not necessarily unitary:
%\vspace{-10pt}
\begin{align*}
\begin{tabular}{ccc}
    \begin{tikzpicture}[yscale=0.8]
                \node (0) at (0, 1) {};
                \node (2) at (-1, 0) {$\mbox{id}_A = $};
                \node (1) at (0.3, 1) {$A$};
                \node (3) at (0, -1) {};
                \node (4) at (0.3, -1) {$A$};
                \draw (3.center) to (0.center);
    \end{tikzpicture}
& \quad
    \begin{tikzpicture}[yscale=0.8]
                \node (0) at (0, 1) {};
                \node (2) at (-1, 0) {$g\circ f = $};
                \node (1) at (0.5, 1) {$C$};
                \node (3) at (0, -1) {};
                \node [style=morphism] (4) at (0, -0.5) {$f$};
                \node (5) at (0.5, -1) {$A$};
                \node (6) at (0.5, 0) {$B$};
                \node [style=morphism] (7) at (0, 0.5) {$g$};
                \draw (3.center) to (4.south);
                \draw (4.north) to (7.south);
                \draw (7.north) to (0.center);
    \end{tikzpicture}
    & \quad
            \begin{tikzpicture}[yscale=0.8]
                \node (0) at (0, 1) {};
                \node (2) at (-1, 0) {$\ket{\psi} = $};
                \node (1) at (0.3, 1) {$A$};
                \node [morphism] (3) at (0, -1) {$\ket{\psi}$};
                \draw (3.north) to (0.center);
    \end{tikzpicture}
\\
    \begin{tikzpicture}[yscale=0.8]
                \node (0) at (0, 1) {};
                \node (2) at (-1.2, 0) {$f_{1}\times f_{2} = $};
                \node (1) at (0.3, 1) {$B$};
                \node (3) at (0, -1) {};
                \node [style=morphism] (4) at (0, 0) {$f_1$};
                \node (5) at (0.3, -1) {$A$};
                \node (6) at (1.25, 1) {};
                \node (7) at (1.55, 1) {$D$};
                \node (8) at (1.25, -1) {};
                \node [style=morphism] (9) at (1.25, 0) {$f_{2}$};
                \node (10) at (1.55, -1) {$C$};
                \draw (3.center) to (0, -0.35);
                \draw (0,0.3) to (0.center);
                \draw (8.center) to (1.25, -0.35);
                \draw (1.25,0.3) to (6.center);
    \end{tikzpicture}
    & \quad \quad
        \begin{tikzpicture}[yscale=0.8]
                \node (0) at (0, 1) {};
                \node (2) at (-2, 0) {$f_{3}:A\to B\times C = $};
                \node (1) at (0.4, 1) {$C$};
                \node at (-0.2, 1) {$B$};
                \node at (0.1, -1) {$A$};
                \node at (0.1,0) {$f_{3}$};
                \node [morphism, xscale=2] (3) at (0.1, 0) {};
                \draw (0.1,-0.75) to (3.south);
                \draw (-0.2,0.75) to (-0.2,0.3);
                \draw (0.4,0.75) to (0.4,0.3);
    \end{tikzpicture}
    &
\\
\end{tabular}
\end{align*}
The state relation is understood where the missing input wire means a map from the ``empty" diagram which is the set $\{\bullet\}$, so that all relations $\ket{\psi}:\{\bullet\}\to A$ give subsets of $A$.\footnote{In any dagger compact category states are morphisms from the monoidal unit, which, in \textbf{Rel} is the singleton \cite{coecke-abramsky-cqm}.}

\begin{definition}
The \emph{adjoint} of a relation $f:A\to B$ is its relational converse $f^{-1}:B\to A$.
\end{definition}

\noindent This is what motivated our definition of unitary relations and is graphically represented by simply flipping the diagram upside down.

Having introduced this notation, we are now able to collect some standard results from the literature~\cite{InteractQOb}, where they are often defined as more general structures. We include these definitions for use in later proofs, and so present them in terms specific and sufficient for our setting.
\begin{definition}
A \textit{comonoid} is a triple \whitecomonoid{A} of a system $A$, a relation  $\tinycomult[whitedot] : A \to A \times A$ called the comultiplication, and a relation $\tinycounit[whitedot] : A \to \{\bullet\}$ called the counit, satisfying coassociativity and counitality equations:
\def\frobscale{0.5}
\begin{calign}
\begin{aligned}
\begin{tikzpicture}[thick, yscale=-0.4, xscale=0.4]
\draw (0,0) to [out=up, in=\swangle] (0.5, 1);
\draw (1,0) to [out=up, in=\seangle] (0.5,1);
\draw (2,0) to [out=up, in=\seangle] (1.25,2);
\draw (0.5,1) to [out=up, in=\swangle] (1.25, 2);
\draw (1.25,2) to (1.25, 3);
\node [whitedot] at (0.5,1) {};
\node [whitedot] at (1.25,2) {};
\end{tikzpicture}
\end{aligned}
\quad=\quad
\begin{aligned}
\begin{tikzpicture}[xscale=-1, thick, yscale=-0.4, xscale=0.4]
\draw (0,0) to [out=up, in=\swangle] (0.5, 1);
\draw (1,0) to [out=up, in=\seangle] (0.5,1);
\draw (2,0) to [out=up, in=\seangle] (1.25,2);
\draw (0.5,1) to [out=up, in=\swangle] (1.25, 2);
\draw (1.25,2) to (1.25, 3);
\node [whitedot] at (0.5,1) {};
\node [whitedot] at (1.25,2) {};
\end{tikzpicture}
\end{aligned}
&\qquad
\begin{aligned}
\begin{tikzpicture}[thick,  yscale=-0.4, xscale=0.4]
\draw (0,-1.5) to (0,-0.5) to [out=up, in=\swangle] (0.75,0.5) node [whitedot] {} to (0.75,1.5);
\draw (1.5,-0.5) node [whitedot] {} to [out=up, in=\seangle] (0.75,0.5);
\end{tikzpicture}
\end{aligned}
\quad=\quad
\begin{aligned}
\begin{tikzpicture}[thick, yscale=-0.4, xscale=0.4]
\draw (0,0) to (0,3);
\end{tikzpicture}
\end{aligned}
\quad=\quad
\begin{aligned}
\begin{tikzpicture}[thick, xscale=-1,  yscale=-0.4, xscale=0.4]
\draw (0,-1.5) to (0,-0.5) to [out=up, in=\swangle] (0.75,0.5) node [whitedot] {} to (0.75,1.5);
\draw (1.5,-0.5) node [whitedot] {} to [out=up, in=\seangle] (0.75,0.5);
\end{tikzpicture}
\end{aligned}
\end{calign}
\end{definition}

\noindent
In equational form, writing $\delta=\tinycomult[whitedot]$ and $\epsilon=\tinycounit[whitedot]$, these are $(\delta\times \mbox{id}_A)\circ\delta = (\mbox{id}_A\times \delta)\circ\delta$ and $\mbox{id}_A\times\epsilon\circ\delta= \mbox{id}_A= \epsilon\times\mbox{id}_A\circ\delta $. Using the relational converse, we can flip the constraining equations upside down to  obtain the associated monoid $(A, \delta^{-1}=\tinymult[whitedot],\, \epsilon^{-1}=\tinyunit[whitedot])$. We can then ask for the comonoid and monoid to interact in various ways.
\begin{definition}
A comonoid \whitecomonoid{A} and its monoid under the relational converse form a \emph{dagger-Frobenius algebra} when the following equation holds:
\begin{equation}\label{eq:frobenius}
\begin{aligned}
\begin{tikzpicture}[thick, yscale=0.4, xscale=0.4]
    \draw (0,0) to (0,1) to [out=up, in=\swangle] (0.5,2) node [whitedot] {} to (0.5,3);
    \draw (0.5,2) to [out=\seangle, in=\nwangle] (1.5,1) node [whitedot] {};
    \draw (1.5,0) to (1.5,1) to [out=\neangle, in=down] (2,2) to (2,3);
\end{tikzpicture}
\end{aligned}
    \quad = \quad
\begin{aligned}
\begin{tikzpicture}[thick, yscale=-0.4, xscale=0.4]
    \draw (0,0) to (0,1) to [out=up, in=\swangle] (0.5,2) node [whitedot] {} to (0.5,3);
    \draw (0.5,2) to [out=\seangle, in=\nwangle] (1.5,1) node [whitedot] {};
    \draw (1.5,0) to (1.5,1) to [out=\neangle, in=down] (2,2) to (2,3);
\end{tikzpicture}
\end{aligned}
\qquad\qquad
\delta^{-1}\times\mbox{id}_A\circ\mbox{id}_A\times\delta=
\mbox{id}_A\times\delta^{-1}\circ\delta\times\mbox{id}_A
  \end{equation}
\end{definition}
\begin{definition}
\label{def:classicalstruct}
A \emph{classical structure} is a dagger-Frobenius algebra \whitecomonoid{A} satisfying the \emph{specialness} \eqref{eq:special} and \emph{symmetry} \eqref{eq:sym} conditions:
\begin{equation}
\label{eq:special}
\begin{aligned}
\begin{tikzpicture}[thick, yscale=0.4, xscale=0.4]
\draw (0,0.25) to (0,1) node [whitedot] {} to [out=\nwangle, in=down] (-0.5,1.5) to [out=up, in=\swangle] (0,2) node [whitedot] {} to (0,2.75);
\draw (0,1) to [out=\neangle, in=down] (0.5,1.5) to [out=up, in=\seangle] (0,2);
\end{tikzpicture}
\end{aligned}
\quad=\quad
  \begin{aligned}
  \begin{tikzpicture}[thick, yscale=0.4, xscale=0.4]
  \draw (-0.5,0) to (-0.5,3);
  \end{tikzpicture}
  \end{aligned}
  \qquad\qquad\qquad
  \delta^{-1}\circ\delta = \mbox{id}_A
  \end{equation}
  
  \vspace{-10pt}
  \begin{equation}
  \label{eq:sym}
\begin{aligned}
\begin{tikzpicture}[yscale=0.4, xscale=0.4]
\draw (0,-0.5) node [whitedot] {} to (0,0.5) node [whitedot] {} to [out=\nwangle, in=down] (-0.5,1.0) to [out=up, in=down] (0.5,2);
\draw (0,0.5) to [out=\neangle, in=down] (0.5,1) to [out=up, in=down] (-0.5,2);
\end{tikzpicture}
\end{aligned}
\quad=\quad
\begin{aligned}
\begin{tikzpicture}[yscale=0.4, xscale=0.4]
\draw (0,-0.5) node [whitedot] {} to (0,0.5) node [whitedot] {} to [out=\nwangle, in=down] (-0.5,1.0) to [out=up, in=down] (-0.5,2);
\draw (0,0.5) to [out=\neangle, in=down] (0.5,1) to [out=up, in=down] (0.5,2);
\end{tikzpicture}
\end{aligned}
  \qquad\qquad\qquad
  \nu\circ\delta\circ\epsilon^{-1} = \delta\circ\epsilon^{-1}
\end{equation}
\noindent where the crossing systems represent the relation that swaps the left and right hand systems, i.e. $\nu:=\{((a,b),(b,a))|a,b\in A\}$. In general, dagger-Frobenius algebras that obey~\eqref{eq:sym} are called \emph{symmetric}.
\end{definition}

As was noted in Lemma \ref{lem:sdfa-rel} these classical structures exactly correspond to groupoids. The map $\tinymult[whitedot]:A\times A\to A$ corresponds exactly to groupoid multiplication defined by Equation \ref{eq:groupoid_mult}. When these classical structures are defined with Hilbert spaces and linear maps instead of sets and relations they exactly correspond to orthonormal bases, as stated in Lemma \ref{lem:sdfa-hilb}.

Complementary classical structures can also be defined graphically. Here we color the maps for two different classical structures differently.
\begin{definition}[Complementarity]
\label{def:complementarity}
Two classical structures \whitecomonoid{A} and \graycomonoid{A} are \emph{complementary} when the following equation holds:
\begin{equation}
\label{eq:complementarity}
 \,\,
\begin{pic}[string, yscale=0.5, xscale=0.5]
\draw (-0.5,0.25) to (-0.5,1) node [graydot] {} to [out=left, in=right] (-1,2) node [graydot] {} to [out=left, in=right] (-1.5,1.5) node [whitedot] {} to [out=left, in=down] (-2,2) to [out=up, in=left] (-0.75,3) node (a) [whitedot] {} to [out=right, in=right] (-0.5,1);
\draw (a.center) to +(0,0.75);
\end{pic}
\quad=\quad\,\,\,
\begin{pic}[string, yscale=0.5, xscale=0.5]
\draw (0,0.25) to (0,1) node [graydot] {};
\draw (0,3) node [whitedot] {} to (0,3.75);
\end{pic}
\end{equation}
\end{definition}
This was shown to correspond to the usual notion of unbiased bases for classical structures in the category of Hilbert spaces and linear maps in~\cite{InteractQOb}.

\section{Unitary Oracles}

In order to model blackbox quantum algorithms in this setting, we must define the oracles themselves.
We do this by building up from an abstract definition of the controlled-not gate in the literature. Let the gray classical structure on a system $A$ be given by a basis $Z=\bigoplus^{|H|}G$ and the white classical structure be a basis $X=\bigoplus^{|G|}H$. The comonoid for the gray dot is then the relation $\tinycomult[graydot]:A\to A\times A$ that for $x,a,b\in H$ is given by
\[ \{(x,(a,b))~|~a\bullet_Zb=x\}. \]

\begin{definition}[\cite{zeng-unitary}]
\label{eq:generalizedcnot}
The abstract controlled-not is given by a composition of the comonoid for Z and the monoid for X:
\begin{equation}
\label{eq:cnot_rel}
\,\,
\begin{aligned}
\begin{tikzpicture}[yscale=0.4, xscale=0.4,string]
\node (b) [graydot] at (0,0) {};
\node (w) [whitedot] at (1,1) {};
\draw (-0.75,2) to [out=down, in=left] (b.center);
\draw (b.center) to [out=right, in=left] (w.center);
\draw (w.center) to (1,2);
\draw (b.center) to (0,-1);
\draw (w.center) to [out=right, in=up] (1.75,-1);
\end{tikzpicture}
\end{aligned}
\qquad \qquad\qquad
\begin{tabular}{l}
\mbox{CNOT:} $H\times H\to H\times H ::$ \\
$\{((x,y),(a,b\circ_Xy))~|~a\bullet_Zb=x\}.$ \\
\end{tabular} 
\end{equation}
\end{definition}
It can be shown that in the traditional quantum setting of Hilbert spaces and linear maps, this exactly corresponds to the usual controlled-not. This also leads to the following useful theorem, which can be abstractly proved.

\begin{theorem}[Complementarity via a unitary \cite{zeng-unitary}]
\label{thm:complementarityunitary}
  Two classical structures are complementary if and only if the abstract controlled-not from Definition \ref{eq:generalizedcnot} is unitary.
\end{theorem}

\noindent This allows us to prove the following about  complementary bases in QCRel.
\begin{theorem}
Two bases (Z and X) in QCRel are complementary, in the sense of Theorem~\ref{thm:compl}, if and only if the relation in \eqref{eq:cnot_rel} is a bijection.
\end{theorem}
\begin{proof}
The relevant relation can clearly be seen to be the composite in Definition~\ref{eq:generalizedcnot} as:
\begin{align}
\{((a,b,y),(a,b\circ_Xy))\} \circ \{((x,y),(a,b,y))~|~a\bullet_Zb=x\}.
\end{align}
Thus the abstract proof of\ Theorem \ref{thm:complementarityunitary} from \cite{zeng-unitary} goes through unchanged.
\end{proof}

An oracle is then introduced as a controlled-not where we have embedded a particular kind of relation that abstractly must be a self-conjugate comonoid homomorphism \cite{zeng-unitary}. We construct such relations in the following lemmas.

\begin{definition}
Let $G$ and $H$ be groupoids with with groupoid multiplications $\bullet_G$ and $\bullet_H$ respectively. Let $\mbox{id}_{G}=\bigcup_{X\in\mbox{Ob}(G)}\mbox{id}_X$ and similarly define $\mbox{id}_{H}$. A \emph{groupoid homomorphism relation} $R:G\to H$ obeys the following condition for $g_1,g_2\in G$:
%two conditions for $g_1,g_2\in G$ and $h_1,h_2\in H$:
\begin{align}
R(g_1\bullet_Gg_2) &= R(g_1)\bullet_HR(g_2) %\\
%R(\mbox{id}_G) &= \mbox{id}_H
\end{align}
\end{definition}
\noindent Note that while this in many ways resembles a groupoid homomorphisms, it is actually a weakening of this notion, in that groupoid homomorphism relations are not required to be total functions and have no explicit requirement on their identity morphisms.

\begin{definition}
A \emph{monoid homomorphism relation} is a monoid homomorphism in the category of sets and relations. Specifically, let $A$ and $B$ be sets equipped with monoids $(A,\tinymult[whitedot],\tinyunit[whitedot])$ and $(B,\tinymult[blackdot],\tinyunit[blackdot])$ respectively. A relation $r:A\to B$ is a monoid homomorphism when it obeys the following two conditions:
\begin{align}
\label{eq:monone}
r\circ\tinymult[whitedot] &= \tinymult[blackdot]\circ(r\times r) 
\end{align}
\begin{align}
\label{eq:monunit}
r\circ\tinyunit[whitedot] &= \tinyunit[blackdot]
\end{align}
A \emph{comonoid homomorphism relation} is defined similarly, using duals of the above conditions.
\end{definition}

\begin{lemma}
\label{lem:mongrphom}
A groupoid homomorphism relation that is surjective on objects is a monoid homomorphism relation.
\end{lemma}
\begin{proof}
Included in Appendix \ref{app:grphom}.
\end{proof}

We then dualize the proof of Lemma \ref{lem:mongrphom} to conclude that:
\begin{lemma}
\label{lem:classicalRelation}
Let $F:H\to G$ be a functor such that $F^{\mbox{\tiny op}}$ is a groupoid homomorphism relation that is surjective on objects. $F$ is a comonoid homomorphism relation.
\end{lemma}
\noindent We call these comonoid homomorphism relations \emph{classical relations}. These are relations that properly preserve the structure of the bases where classical data is embedded.  In the quantum case they take basis elements to basis elements. Some examples in QCRel are listed in Appendix \ref{app:clRel}. In order to define unitary oracles, we also need these relations to be self-conjugate:

\begin{definition}[\cite{zeng-unitary}]
In a monoidal dagger-category, a comonoid homomorphism $f:\blackcomonoid{A} \to \graycomonoid{B}$ between dagger-Frobenius comonoids is \emph{self-conjugate} when the following property holds:
\begin{equation}
\label{eq:comonoidhomomorphismselfconjugate}
\begin{aligned}
\begin{tikzpicture}[yscale=0.5, xscale=0.5, thick]
\node [morphism] (f) at (2,1) {$f$};
\draw (0,-1) to [out=up, in=left, in looseness=0.9] (1,2) node [graydot] {} to (1,2.5) node [graydot] {};
\draw (1,2) to [out=right, in=up] (f.north);
\draw (f.south) to [out=down, in=left] (3,0) node [blackdot] {} to [out=right, in=down, out looseness=0.9] (4,3);
\draw (3,0) to (3,-0.5) node [blackdot] {};
\node [graydot] at (1,2) {};
\end{tikzpicture}
\end{aligned}
\quad=\quad
\begin{aligned}
\begin{tikzpicture}[yscale=0.6, xscale=0.5, thick]
\node (f) at (0,0) [morphism] {$f^{\dagger}$};
\draw (0,-1.5) to (f.south);
\draw (f.north) to (0,1.5);
\end{tikzpicture}
\end{aligned}
\end{equation}
\end{definition}
The meaning of this equation in relations is explicated in the following lemma.

\begin{lemma}
All classical relations $f:Z^A\to Z^B$ between groupoids $Z^A=\bigoplus^NG^A$ and $Z^B=\bigoplus^{N'}G^B$ are self-conjugate.
\end{lemma}
\begin{proof}
In QCRel, our dagger-Frobenius structures are groupoids and, if they are complementary to some other groupoid, then they are of the form $Z^A=\bigoplus^NG$ and $Z^B=\bigoplus^{N'}H$. We annotate the definition of self-conjugacy for some arbitrary element $(g,n)$, the element $g$ from the $n$-th group. Recall from Section~\ref{sec:model} that $f^{\dagger}=f^{-1}$ in QCRel.
\begin{equation}
\begin{aligned}
\begin{tikzpicture}[yscale=0.7, xscale=1.1, thick]
\node [morphism] (f) at (2,1) {$f$};
\draw (0,-1) to [out=up, in=left, in looseness=0.9] (1,2) node [graydot] {} to (1,3.3) node [graydot] {};
\draw (1,2) to [out=right, in=up] (f.north);
\draw (f.south) to [out=down, in=left] (3,-0.7) node [blackdot] {} to [out=right, in=down, out looseness=0.9] (4,3);
\draw (3,-0.7) to (3,-2) node [blackdot] {};
\node [graydot] at (1,2) {};
\node at (0,-1.25) {\small $(g,n)$};
\node at (-0.25,1.5) {\small $(g,n)$};
\node at (0.9,2.65) {\small $\{(id_G,j) | 1 \leq j \leq N\}$};
\node at (2.5,1.7) {\small $(g^{-1},n)$};
\node at (2.1,-0.15) {\small $f^{-1}(g^{-1},n)$};
\node at (4.3,0) {\small $\left[f^{-1}(g^{-1},n)\right]^{-1}$};
\node at (2.95,-1.4) {\small $\{(id_H,k) | 1 \leq k \leq N'\}$};
\end{tikzpicture}
\end{aligned}
\quad=\quad
\begin{aligned}
\begin{tikzpicture}[yscale=0.9, xscale=1.1, thick]
\node (f) at (0,0) [morphism] {$f^{-1}$};
\draw (0,-1.5) to (f.south);
\draw (f.north) to (0,1.5);
\node at (0,-2) {\small $(g,n)$};
\node at (0,2) {\small $f^{-1}(g,n)$};
\end{tikzpicture}
\end{aligned}
\end{equation}
Thus, a relation $f$ is self-conjugate if and only if for all elements $(g,n)$ it is the case that $[f^{-1}(g^{-1},n)]^{-1}=f^{-1}(g,n)$. From Lemma \ref{lem:classicalRelation} the converse of the classical relation $f$ is a monoid homomorphism relation whose multiplication is the groupoid operation
\end{proof}

Classical relations, as self-conjugate comonoid homomorphisms, lead to unitary oracles.

\begin{definition}[Oracle \cite{zeng-unitary}]
\label{oracle}
Given a groupoid $Z^A:\blackcomonoid{A}$, a pair of complementary groupoids $Z^B:\graycomonoid{B}$ and $X^B:\whitecomonoid{B}$, and a classical relation $R : \blackcomonoid{A} \to \graycomonoid{B}$, an \emph{oracle} is defined to be the following endomorphism of $A \times B$:

\begin{tabularx}{\linewidth}{XX}
\[\begin{aligned}
\begin{tikzpicture}[string,xscale=0.6, yscale=0.4]
    \node (dot) [blackdot] at (0,1) {};
    \node (f) [morphism] at (0.7,2) {$R$};
    \node (m) [whitedot] at (1.4,3) {};
\draw (0,0.25)
        node [below] {$A$}%
    to (0,1)
    to [out=left, in=south] (-0.7,2)
    to (-0.7,3.75)
        node [above] {$A$};
\draw (0,1)
    to [out=right, in=south] (f.south);
\draw  (f.north)
    to [out=up, in=left] (1.4,3)
    to [out=right, in=up] +(0.7,-1)
    to (2.1,0.25)
        node [below] {$B$};;
\draw (m.center) to +(0,0.75) node [above] {$B$};
\end{tikzpicture}
\end{aligned}\]
& {\begin{align*}
&\mbox{\emph{OracleRel(R)}}:A\times B\to A\times B  ::\\
&\{((x,y),(a,c\circ_Xy))~|~ \\ &\hspace{50pt}\exists b\in A, s.t.~a\bullet_{Z^A}b=x\mbox{\emph{ and }} bRc\}.
\end{align*}}
\end{tabularx}
\end{definition}
\begin{theorem}
\label{thm:familyofunitaries}
Oracles are unitary.
\end{theorem}
\begin{proof}
Proved in the abstract setting for Definition \ref{oracle} in \cite{zeng-unitary}, when $R$ is a self-conjugate comonoid homomorphism.  Though there are others, classical relations $R$ are necessary and sufficient in our cases as the algorithms that follow additionally require that the comonoids be part of classical structures.
\end{proof}

\begin{corollary}
OracleRel is a bijection.
\end{corollary}
\begin{proof}
This follows directly from Theorem \ref{thm:familyofunitaries} and Corollary \ref{cor:bijections}.
\end{proof}

\section{The Fourier Transform in Relations}
\label{sec:FT}

In these algorithms we use the quantum Fourier transform for relations \cite{strongCompFT}. This is a generalized quantum Fourier transform whose definition is motivated through the relationship between classical and unbiased states of two bases.  For abelian groups $G$ and $H$, consider two groupoids $Z=\bigoplus^{|H|}G$ and $X=\bigoplus^{|G|}H$ to be complementary bases of the same system.

\begin{definition}
\label{def:FTRel}
The \emph{quantum Fourier transform in relations} corresponds to preparing classical states of $Z$ and measuring them against classical states of $X$.
%is an isomorphism from the classical states of $Z$ to the unbiased states %of $X$, i.e.
%\begin{align*}
%\{G_h\}\mapsto \{h_g|\forall g\in G\}
%\end{align*}
\end{definition}

\begin{example}
Take $G=\mathbb{Z}_2=\{0,1\}$, $H=\mathbb{Z}_1=\{\star\}$, $Z = G$ and $X=H\oplus H = \{ (\star,0),(\star,1) \}$. The computational basis is the family $\ket{H_g}_{g\in G}$ of classical states for $X$, i.e. $H_0 = \{(\star,0)\}$ and $H_1 = \{(\star,1)\}$. The quantum Fourier basis is a single classical state $G_\star = \{(\star,0), (\star,1)\}$ for $Z$. In this case all states can be prepared in the computational basis, but  measurement in the quantum Fourier basis is trivial.
\end{example}

\begin{example}
Take $G=\mathbb{Z}_2=\{0,1\}$, $H=\mathbb{Z}_2=\{a,b\}$, $Z = G \oplus G = \{ (0,0),(1,0),(0,1),(1,0)\}$ and $X= H \oplus H = \{ (a,a), (b,a), (a,b), (b,b) \}$. The computational basis is the family $\ket{H_g}_{g \in G}$ of classical states for $X$, i.e. $H_0 = \{(a,a),(b,a)\}$ and $H_1 = \{(a,b),(b,b)\}$. The quantum Fourier basis is the family $\ket{G_h}_{h\in H}$ of classical states for $Z$, i.e. $G_a = \{(0,0),(0,1)\}$ and $G_b = \{(1,0),(1,1)\}$.
\end{example}

See \cite{strongCompFT} to fully motivate this definition of the Fourier transform in QCRel and for its relationship to the usual Hadamard and Fourier transforms for Hilbert spaces and linear maps.

\section{The Deutsch-Jozsa Algorithm in QCRel}

The well known Deutsch-Jozsa algorithm is an early quantum algorithm that demonstrates a speedup over exact classical computation \cite{DJAlg1992}. It takes as input a function promised to be either constant or balanced and returns which, deterministically using only a single oracle query. In this section, we model the algorithm's steps in QCRel just as it is implemented with Hilbert spaces and linear maps. This approach is somewhat dual to the usual one where different algorithms are compared on the same problem. Here we run the same abstract protocol (implemented in a different model) with the same query complexity and compare the different problems that it solves.

To run this algorithm in QCRel we use two systems.  System $A$ has cardinality $n$ and system $B$ has cardinality $\ge 2$. Take $Z^A=\bigoplus^{|H^{A}|}G^A$ and $X^A=\bigoplus^{|G^{A}|}H^A$ to be complementary bases of $A$. Take $Z^B=\bigoplus^{|H^{B}|}G^B$ and $X^B=\bigoplus^{|G^{B}|}H^B$ to be complementary bases of $B$, such that $X^B$ has at least two classical states. In analogy with the usual specification, the algorithm proceeds with the following steps.
\begin{enumerate}
\item Prepare $A$ in the zero state $|G^{A}_0\rangle$. Prepare $B$ in the state given by the second classical state of $Z^B$, i.e. $|G^B_1\rangle$.

\item Apply the Fourier transform, as given by Definition \ref{def:FTRel}, to each system, resulting in states $\ket{H_0^A}$ and $\ket{H_1^B}$ respectively.

\item Apply an oracle (Definition~\ref{oracle}), built from a classical relation $f:Z^A\to Z^B$.

\item Again apply the Fourier transform to system $A$ and then measure it in the $Z$ basis.
\end{enumerate}

\noindent This sequence of steps is an instance in sets and relations of the abstract Deutsch-Jozsa algorithm from \cite{vicary-tqa}, which translates to the following relation where we have already applied the Fourier transform to the input and output systems:

\begin{equation}
\label{eq:reldj}
\begin{aligned}
\begin{tikzpicture}[string, yscale=0.7, xscale=0.8]
    \node (dot) [blackdot] at (0,1) {};
    \node (f) [morphism] at (0.7,2) {$f$};
    \node (m) [whitedot] at (1.4,3) {};
\draw (0,0.25)
        node [blackdot] {}
    to (0,1)
    to [out=left, in=270] (-0.7,2)
    to (-0.7,3.75)
        node [blackdot] {};
\draw (0,1)
    to [out=right, in=270] (f.south);
\draw  (f.north)
    to [out=up, in=left] (1.4,3)
    to [out=right, in=up] +(0.7,-1)
    to (2.1,0.25)
        node [graydot] {};
\node at (2.15,0.25) [anchor=west] {$|H^B_1\rangle$};
\node at (0.05,0.25) [anchor=west] {$|H^A_0\rangle$};
\node at (-0.65,3.75) [anchor=west] {$\langle H^A_0|$};
\draw (m.center) to +(0,0.75)
        node [above] {};
\draw [thin, dashed] (-1.25,0.7) to (7.5,0.7);
\draw [thin, dashed] (-1.25,3.3) to (7.5,3.3);
\node at (3.5,0) [anchor=west] {\small Prepare initial states and apply FT};
\node at (3.5,2) [anchor=west] {\small Apply a unitary map};
\node at (3.5,4) [anchor=west] {\small Apply FT and measure the first system};
\end{tikzpicture}
\end{aligned}
\end{equation}
that is explicitly written as:
\begin{align*}
\mbox{DJAlg}(f)&::\{\bullet\}\times \{\bullet\} \to \{\bullet\}\times B \\
&=
\bra{H_0^A}\times{\mbox{id}_B}\circ\mbox{OracleRel}(f)\circ\ket{H^A_0}\times\ket{H^B_1}
\\ &=
\{((\bullet,\bullet),(\bullet,z))\;|\; 
  z\in H_1^B \mbox{ and } \exists y\in H_0^A, \mbox{ s.t. }yfz\}.
\end{align*}

\begin{theorem}[\cite{vicary-tqa}]
\label{def:bc}
In any dagger compact category with complementary bases, the algorithm in Equation \ref{eq:reldj} will, with a single oracle query, distinguish \emph{constant} and \emph{balanced} classical relations $f:Z^A\to Z^B$ according to the following abstract definitions. Here $\ket{x}$ is a classical point of $Z^A$ and the zero scalar $0$ is, in \cat{Rel}, the empty relation:
\begin{equation}
\label{eq:bc}
\mbox{\\ constant}:\quad
\begin{aligned}
\begin{tikzpicture}[scale=0.8]
\node (f) [morphism] at (0,0) {$f$};
\draw (0,-1) to (f.south);
\draw (f.north) to (0,1);
\end{tikzpicture}
\end{aligned}
\;=\;
\begin{aligned}
\begin{tikzpicture}[scale=0.8]
\draw (0,-1) to (0,-.4)
    node [blackdot] {};
\draw (0,0.5) node [state] {$x$} to (0,1);
\end{tikzpicture}
\end{aligned}
=\ket{x}\circ\tinycounit[blackdot]
\qquad\qquad\qquad \mbox{balanced:\quad}
\begin{aligned}
\begin{tikzpicture}[string, scale=0.8]
\node [morphism] (f) at (0,0) {$f$};
\draw (0,-0.85) node [blackdot] {} to (f.south);
\draw (f.north) to (0,0.75) node [graydot, hflip] {};
\node at (0.05,0.75) [anchor=west] {$2$};
\end{tikzpicture}
\end{aligned}
\quad=\quad
0, \vspace{-5pt}
\end{equation}
where 
\begin{tikzpicture}[string, yscale=0.75]
\draw (f.north) to (0,0.75) node [graydot, hflip] {};
\node at (0.05,0.75) [anchor=west] {$2$};
\end{tikzpicture} is the dagger adjoint of the second classical state of $X^B$.
\end{theorem}
That these definitions coincide with the usual ones for constant and balanced functions is shown in \cite{vicary-tqa}. In QCRel, the effect 
\begin{tikzpicture}[string,scale=0.75]
\draw (f.north) to (0,0.75) node [blackdot, hflip] {};
%\node at (0.05,0.75) [anchor=west] {$2$};
\end{tikzpicture} is $\langle H^A_1|$, which acts as a measurement of system $A$ after applying the oracle.
We illustrate the details of the QCRel model of this algorithm by example and then with general definitions.

\begin{example}
Take $A=\{0,1,2,3\}$ and $B=\{a,b,c,d\}$ to be four element systems. We define complementary bases on these systems as the following:
\begin{align*}
\begin{tabular}{|c|c|}\hline
System $A$ & System $B$ \\\hline
$Z^A = \mathbb{Z}_2\oplus\mathbb{Z}_2 \mbox{~~s.t.}$ & $Z^B = \mathbb{Z}_2\oplus\mathbb{Z}_2\mbox{~~s.t.}$ \\
$G_0^A=\{0,1\},G_1^A=\{2,3\}$ & $G_0^B=\{a,b\},G_1^B=\{c,d\}$ \\ \hline
$X^A = \mathbb{Z}_2\oplus\mathbb{Z}_2 \mbox{~~s.t.}$ & $X^B = \mathbb{Z}_2\oplus\mathbb{Z}_2\mbox{~~s.t.}$ \\
$H_0^A=\{0,2\},H_1^A=\{1,3\}$ & $H_0^B=\{a,c\},H_1^B=\{b,d\}$ \\ \hline
\end{tabular}
\end{align*}

From Equation \ref{eq:bc}, we then define constant and balanced classical relations using the following dictionary:
\begin{align}
\begin{aligned}
\begin{tikzpicture}[string]
\draw (0,0.2) to (0,0.75) node [blackdot, hflip] {};
\node at (0.05,0.75) [anchor=west] {};
\end{tikzpicture}
\end{aligned}
&\quad= \{(0,\bullet),(2, \bullet)\},  \quad \mbox{the adjoint of the first classical state of } X^A \\
\begin{aligned}
\begin{tikzpicture}[string]
\node (x) [state, xscale=0.75] at (0,0) {$x$};
\draw (0,0) to (0,0.5) {};
\end{tikzpicture}
\end{aligned}
&\quad= \{(\bullet,a),(\bullet,b)\}\mbox{ OR }\{(\bullet,c),(\bullet,d)\},  \quad \mbox{a classical state of }Z^B \\
\begin{aligned}
\begin{tikzpicture}[string]
\draw (0,0.2) to (0,0.75) node [graydot, hflip] {};
\node at (0.05,0.75) [anchor=west] {$2$};
\end{tikzpicture}
\end{aligned}
&\quad=\{(b,\bullet),(d, \bullet)\},  \quad \mbox{ adjoint of the second classical state of } X^B \\
\begin{aligned}
\begin{tikzpicture}[string]
\node (x) [blackdot] at (0,0) {};
\draw (0,0) to (0,0.5) {};
\end{tikzpicture}
\end{aligned}
&\quad= \{(\bullet,0),(\bullet,2)\},  \quad \mbox{the first classical state of } X^A
\end{align}

Thus there are two constant classical relations\footnote{A list of more example classical relations is given in Appendix \ref{app:clRel}.} $f:Z^A\to Z^B$, one for each classical state of $Z^B$. They are:
\begin{align*}
\{ (0,a) ,(0,b), (2,a), (2,b) \} \quad \mbox{and} \quad
\{ (0,c) ,(0,d), (2,c), (2,d) \}.
\end{align*}
By Theorem~\ref{def:bc}, balanced classical relations are those which do not relate $0$ or $2$ to either $b$ or $d$. There are four balanced classical relations for this example:
\begin{align*}
\begin{tabular}{cc}
\{(0,c),(2,c),(1,d),(3,d)\} &\qquad \{(0,a),(1,b),(2,c),(3,d)\} \\
\{(2,a),(3,b),(0,c),(1,d)\} &\qquad \{(0,a),(2,a),(1,b),(3,b)\} \\
\end{tabular}
\end{align*}
For a classical relation promised to be in one of these two classes, we can distinguish which with a single oracle query.
\end{example}

We generalize these definitions of constant and balanced classical relations to the following:
\begin{definition}
\label{def:const}
Let $Z^A=\oplus^NG_i$. A \emph{constant relation} $f:Z^{A}\to Z^{B}$ relates all id$_{G_i}$ to a single classical state of $Z^B$.
\end{definition}

\begin{definition}
\label{def:balanced}
A relation $f:Z^{A}\to Z^{B}$ is \emph{balanced} when no element in the first classical state of $X^{A}$ is related to an element in the second classical state of $X^{B}$.
\end{definition}

\begin{theorem}
\label{thm:dj_speedup}
The Deutsch-Jozsa algorithm defined above distinguishes constant relations from balanced relations in a single oracle query.
\end{theorem}
\begin{proof}
This follows immediately from the abstract proof of the Deutsch-Jozsa algorithm in \cite{vicary-tqa}.
\end{proof}

This result shows that we are able to model the Deutsch-Jozsa algorithm in the nondeterministic classical setting of QCRel.
\vspace{-10pt}
\section{Single-shot Grover's Algorithm}

The usual Grover's algorithm~\cite{grover1996fast} takes as input a set $S$ and an indicator function $f:S\to\{0,1\}$ and outputs an element $s\in S$ such that $f(s)=1$. Though the algorithm is usually probabilistic and runs a repeated series of ``Grover steps", here we consider the deterministic version that runs with a single step. In this section we will consider the generalization of the single-shot Grover algorithm where the codomain of the indicator function is allowed to be an arbitrary group~\cite{vicary-tqa}. Our setup requires the set $S$, as one system, as well as another system $B$. We define the basis $Z^{S}=\bigoplus^{|H^S|}G^S$ and $X^S=\bigoplus^{|G^S|}H^S$ on the $S$ system.  System $B$ has complementary bases $Z^B=\bigoplus^{|H^B|}G^B$ and $X^B=\bigoplus^{|G^B|}H^B$. Let $\ket{\sigma}$ be the first classical state of $X^B$, e.g. is $X^B=\mathbb{Z}_2\oplus\mathbb{Z}_2$ then $\ket{\sigma}=\{(\star,1),(\star,3)\}$, where $1$ and $3$ are the non-identity elements of that factors of $X^B$. Let $\bra{\rho}$ be the converse of a classical state of $X^S$. Recall that $\ket{G^S_0}=\{(\star,g)\,|\,g\in G^S\mbox{ is the first factor group of } Z^S\}$ is a classical point of $Z^S$, and that, by the complementary relationship of classical and unbiased points (Section~\ref{sec:FT}), $\ket{H^S_0}\cong\{(\star,\idm{G^S_i})\,|\,G^S_i\mbox{ is a factor group of } Z^S\}$.

In QCRel, the algorithm proceeds by the following steps:
\begin{enumerate}
\item Prepare system $S$ in the state $\ket{G_0}$ and system $B$ in the state $\ket{\sigma} = \ket{0\vee 1}$.

\item Apply the Fourier transform to system $S$, resulting in state $\ket{H_0}$.

\item Apply the oracle for a classical indicator relation $f:Z^S\to Z^B$.

\item Apply a diffusion relation $D:S\to S$ to system $S$.

\item Measure system $S$ in the $X^S$ basis.

\end{enumerate}

The diagrammatic presentation for this procedure from \cite{vicary-tqa} is:
\begin{equation}
\label{eq:grovertopological}
\begin{aligned}
\begin{tikzpicture}[thick, xscale=0.7, yscale=0.35]
\begin{pgfonlayer}{foreground}
    \node (f) [smallbox, anchor=south, thick] at (0.7,2) {$f$};
\end{pgfonlayer}
    \node (dot) [blackdot] at (0,1) {};
    \node (m) [whitedot] at ([xshift=0.7cm, yshift=1cm] f.north) {};
\draw (0,-0.25)
        node [blackdot] (bdot) {}
    to (0,1)
    to [out=\nwangle, in=south] (-0.7,2)
    to ([yshift=1.4cm] m.center -| -0.7,1)
        node (rho) [state,hflip,yscale=1.5, xscale=1.1] {};
\node at (-0.7,6.3) {$\bra{\rho}$};
\draw (0,1)
    to [out=\neangle, in=south] (f.south)
    to (f.north)
    to [out=up, in=\swangle] +(0.7,1)
    to [out=\seangle, in=up] +(0.7,-1)
    to (2.1,-0.25)
        node [state,yscale=1.5, xscale=1.1] (sigmadag) {};
\node at (2.1,-1) {$\ket{\sigma}$};        
\draw (m.center) to (1.4,5.75)
        node [above] {};
\node [smallbox] at (-0.7,3.5) {$D$};
\node at (3.5,-0.25) [anchor=west] {Preparation};
\draw [thin, dashed] (-2,0.6) to (7,0.6);
\node at (3.5,2.25) [anchor=west] {Dynamics};
\draw [thin, dashed] (-2,4.5) to (7,4.5);
\node at (3.5,5) [anchor=west] {Measurement};
\end{tikzpicture}
\end{aligned}
\end{equation}
\vspace{-8pt}

\noindent where numerical scalars have been dropped as there is only one non-zero scalar in QCRel.
\footnote{Recall that scalars in a monoidal category with identity object $I$ are maps $s:I\to I$. Thus in \cat{Rel} $I=\{\star\}$, so the only scalars are the empty relation and the identity relation on the singleton set.} 
Recall that $\tinyunit[blackdot]:\{\star\}\to S$ relates the singleton to the elements of $H_0$ and that $\tinycounit[blackdot]$ is its relational converse. We will use the map $\tinyunit[blackdot]\circ\tinycounit[blackdot]:S\to S$ in the following definition. Here there is a special relation $D:S\to S$ called the diffusion operator and defined abstractly in \cite{vicary-tqa}:
\vspace{-10pt}
\begin{equation}
\label{eq:difftopological}
\begin{aligned}
\begin{tikzpicture}[thick, scale=0.5]
\draw (0,0) node [below] {$S$} to node [smallbox] {$D$} (0,2) node [above] {$S$};
\end{tikzpicture}
\end{aligned}
\hspace{5pt}\;=\;\hspace{5pt}
\hspace{-5pt}
\begin{aligned}
\begin{tikzpicture}[thick, scale=0.5]
\draw (0,0) node [below] {$S$} to (0,2) node [above] {$S$};
\end{tikzpicture}
\end{aligned}
\;-\;
\begin{aligned}
\begin{tikzpicture}[thick, scale=0.5]
\draw (0,0) node [below] {$S$} to (0,0.7) node [blackdot] {};
\draw (0,2) node [above] {$S$} to (0,1.3) node [blackdot] {};
\end{tikzpicture}
\end{aligned}
\vspace{-10pt}
\qquad\qquad \qquad D := \{(x,x)\,|\,x\in S\}\bigtriangleup (H_0\times H_0)
\end{equation}
where the subtraction of two relations is given by the symmetric difference of their images. Explicitly then, the relational model for Grover's algorithm is:
\begin{align*}
\mbox{Grover}(f)&:\{\bullet\}\times \{\bullet\} \to \{\bullet\}\times B \\
&=
\bra{\rho}\times{\mbox{id}_B}\circ D\times{\mbox{id}_B}\circ\mbox{OracleRel}(f)\circ\ket{H^S_0}\times\ket{\sigma}
\\ &= \{((\bullet,\bullet),(\bullet,c\circ_{X}x))\;|\; \\
&\hspace{40pt}
 x\in\sigma,y\in\rho, \mbox{id}_{G_n},b,z\in S \mbox{ s.t. } z\bullet_{Z^S} b=\mbox{id}_{G_{n}} \mbox{ and } bfc,zDy\}
\end{align*}

\begin{theorem}
\label{thm:topgrovers}
Equation \ref{eq:grovertopological} is zero only for classical states of $X^S$ denoted $\ket{\rho}$ that satisfy the following equation:
\begin{equation}
\label{eq:zerocondition}
\sigma \circ f \circ \rho = \sigma \circ f \circ \ket{G_0}
\end{equation}
%\begin{equation}
%\label{eq:zerocondition}
%\begin{aligned}
%\begin{tikzpicture}[thick, scale=0.5]
%\draw (0,0) node [smallbox] {$\rho$} to (0,1.25) node [smallbox] {$R$} to %(0,2.5) node [smallbox] {$\sigma$};
%\end{tikzpicture}
%\end{aligned}
%\quad=\quad
%\begin{aligned}
%\begin{tikzpicture}[thick, scale=0.5]
%\draw (0,0) node [blackdot] {} to (0,1.25) node [smallbox] {$R$} to (0,2.5) %node [smallbox] {$\sigma$};
%\node [smallbox, draw=white, fill=none] at (0,0) {};
%\draw (0,0) to (0,0.5);
%\end{tikzpicture}
%\end{aligned}
%\end{equation}
\end{theorem}
\begin{proof}
Proven in \cite{vicary-tqa}. See Section 3.2 equation (34).
\end{proof}

Here $\ket{\sigma}$ is, in general, any fixed classical state of $X^B$. This allows a generalization of the single-shot Grover's algorithm where the cardinality of system $B$ is increased as investigated in \cite{vicary-tqa}.
Consequently, the LHS of Equation \ref{eq:zerocondition} tests if any element in the classical state $\ket{\rho}$ is related to any of the elements in $\ket{\sigma}$. The RHS tests if any of the elements of $G_0$ are related to $\ket{\sigma}$.

\begin{proposition}
The QCRel single-shot Grover algorithm only returns states $\ket{\rho}$ such that for all $h \in H^S_0$, $s\in\rho$ and $x\in \sigma$
\begin{align*}
h f x \quad = \quad \neg(s f x) .
\end{align*}
In other words, the only elements can be possibilistically measured (via the QCRel Born rule in Axiom~\ref{ax:born}) are elements of $S$ that have the opposite mapping to $\sigma$, under the relation $f$, than elements of $H^{S}_0$.
\end{proposition}
\begin{proof}
Theorem \ref{thm:topgrovers} gives an abstract proof and this proposition can be seen to instantiate it by the definitions given here.
\end{proof}

\begin{example}
Let $S=\{0,1,2,3\}$ and choose $Z^S=\mathbb{Z}_2\oplus\mathbb{Z}_2$ and $X^S=\mathbb{Z}_2\oplus\mathbb{Z}_2$ as $G$ (black) and $H$ (white) bases respectively, so that $G_0=\{0,1\}$ and $H_0 = \{0,2\}$. Let $B$ be the four element system with the same bases and choose $\ket{\sigma}=\ket{1\vee3}$. The diffusion operator is then given by
\begin{align*}
D &:= \{(0,0),(1,1),(2,2),(3,3)\}-\{(0,0),(0,2),(2,0),(2,2)\}
\\&=\{(1,1),(3,3),(0,2),(2,0)\}.
\end{align*}

In this case, $D$ happens to be a bijection, it is a unitary relation and thus a possible evolution in QCRel.\footnote{This will not be the case whenever $S$ has more than two factor groups. Unitarity is a stringent condition on processes in QCRel.} Let $f$ be the classical relation\footnote{See Appendix \ref{app:clRel} for a list of classical relations $\mathbb{Z}_2\oplus\mathbb{Z}_2\to\mathbb{Z}_2\oplus\mathbb{Z}_2$.} $\{(0,2),(2,2),(1,3),(3,3)\}$, where elements of $H^S_0$ are not related to elements of $\ket{\sigma}$. Thus the above algorithm will only return classical states of $X^{S}$ that \textit{are} related, under $f$,\ to $\ket{\sigma}$.  The only possible outcome state is $\ket{1\vee 3}$.
\end{example}

\begin{example}
This is the same as the above example, but take $f$ to be the classical relation $\{(0,0),(2,0),(0,1),(2,1)\}$. As an element of $H^S_0$ is related to $\ket{\sigma}$, the algorithm will return classical states of $X^S$ which are \emph{not} mapped to $\ket{\sigma}$, i.e. the state $\ket{1\vee3}$.
\end{example}

\section{The Groupoid Homomorphism Promise Algorithm}

This section models the group homomorphism algorithm from \cite{zeng-unitary} in QCRel.  The quantum version of the algorithm, which operates in $\cat{FHilb}$, takes as input a blackbox function $f:G\to A$ promised to be one of the homomorphisms between group $G$ and abelian group $A$.  It then outputs the identity of the homomorphism. In that paper the full identification algorithm is built up by multiple calls to an instance of the problem for cyclic groups.\footnote{Making use of the structure theorem for abelian groups to complete the general case.} It is this cyclic group subroutine that we consider here. In the relational setting we will move from groups to groupoids. Let groupoid $H$ be complementary to groupoid $G$ and groupoid $B$ be complementary to groupoid $A$. The QCRel GroupHomID algorithm then takes as input a groupoid isomorphism $f:G\to A$.   Let $\ket{\rho}$ be a classical states of $H$, and $\ket{\sigma}$ be a classical state of $B$.

The algorithm has the following abstract specification~\cite{zeng-unitary}:
\begin{align}
\label{eq:groupdIDAlg}
\begin{aligned}
\begin{tikzpicture}[string, scale=0.5, xscale=1.3]
    \node (dot) [blackdot] at (0,1) {};
    \node (f) [morphism] at (0.7,2) {$f$};
    \node (m) [whitedot] at (1.4,3) {};
    \node (topsig) [state, hflip] at (-0.7,3.6) {$\ket{\rho}$};
\draw (0,0)
        node [blackdot] {}
    to (0,1)
    to [out=left, in=south] (-0.7,2)
    to (topsig);
\draw (0,1)
    to [out=right, in=south] (f.south);
\draw  (f.north)
    to [out=up, in=left] (1.4,3)
    to [out=right, in=up] +(0.7,-1)
    to (2.1,0.4)
        node [state] {$\ket{\sigma}$};
\draw (m.center) to (1.4,4.4)
        node [above] {};
\draw [thin, dashed] (-1.25,0.7) to (7.5,0.7);
\draw [thin, dashed] (-1.25,3.3) to (7.5,3.3);
\node at (5,0) [anchor=west] {Prepare initial states};
\node at (5,2) [anchor=west] {Apply a unitary map};
\node at (5,4) [anchor=west] {Measure the left system};
\end{tikzpicture}
\end{aligned}
\end{align}
Let the factor groups of a groupoid $G$ be denoted $G_n$. This gives the following relational model for the algorithm:
\begin{align*}
\mbox{GroupHomID}(f)&:\{\bullet\}\times \{\bullet\} \to \{\bullet\}\times B \\
&=
\bra{\rho}\times{\mbox{id}_B}\circ \mbox{OracleRel}(f)\circ\ket{H_0}\times\ket{\sigma}
\\ &= \{((\bullet,\bullet),(\bullet,c\circ_{X}x))\;|\; \\
&\hspace{40pt}
x\in\sigma,y\in\rho, \mbox{id}_{G_n},b\in A \mbox{ s.t. } y\bullet_G b=\mbox{id}_{G_{n}} \mbox{ and } bfc\}
\end{align*}

\begin{theorem}
The algorithm defined by~\eqref{eq:groupdIDAlg} has output state $\ket{\rho}$ only when for some $x\in \rho$ and some $y\in \sigma$ we have $(y,x)\in f$.
\end{theorem}
\begin{proof}
The verification in \cite{zeng-unitary} simplifies the algorithm in Equation \ref{eq:groupdIDAlg} to:
%\vspace{-7pt}
\begin{equation}
\begin{aligned}
\begin{tikzpicture}[string, scale=0.9]
\node (r) [state] at (0,2) {$\sigma$};
\node (s) [state, hflip] at (0,3.5) {$\rho$};
\node (f) [morphism] at (0,2.75) {$f^{-1}$};
\node (r2) at (1.5,3) [state] {$\sigma$};
\draw (s) to (f.north);
\draw (f.south) to (r);
\draw (r2) to +(0,0.5);
\end{tikzpicture}
\end{aligned}
\qquad= \quad \{(\bullet,(\bullet,x))\;|\;x f^{-1} y\;\mbox{for some } y\in\rho\},
\end{equation}
where we see that post-selection on the left hand system implies the theorem's condition.
\end{proof}

\begin{theorem}
If $f$ is a groupoid isomorphism then the algorithm in Equation \ref{eq:groupdIDAlg} returns all states.
\end{theorem}
\begin{proof}
Groupoid isomorphisms relate every element of the domain to some element in the codomain and relate every element of the codomain to some element of the domain.
\end{proof}

\noindent Still, we can imagine running the algorithm from \eqref{eq:groupdIDAlg} where any classical relation $f$ is allowed as input to obtain non-trivial outcomes.

\begin{appendix}
\section{Appendix: Proof of Lemma \ref{lem:mongrphom}}
\label{app:grphom}
\begin{proof}
Throughout this proof we refer to a groupoid as a category where the elements of the groupoid are the morphisms.  From this perspective a group is a groupoid with a single object. Consider a groupoid homomorphism relation $R:G\to H$ on objects $X,A,B$ of $G$ and morphisms $f$ of $G$.
In order to show that $R$ is a monoid homomorphism relation we first show that it preserves the unit~\eqref{eq:monunit}. We have $R(\bigcup_{X\in\mbox{Ob}(G)} \mbox{id}_X) = \bigcup_{Y\in\mbox{Ob}(H)} \mbox{id}_Y$. Recall that for a set $A$,  $R(A)=\bigcup_{a\in A} R(a)$. It is that case that
\begin{align}
R(\bigcup\nolimits_{X\in\mbox{Ob}(G)} \mbox{id}_X) &= \bigcup\nolimits_{X\in\mbox{Ob}(G)}R(\mbox{id}_X) = \bigcup\nolimits_{X\in\mbox{Ob}(G)}\mbox{id}_{R(X)} \quad \mbox{def. of group hom. rel.} \\
&= \bigcup\nolimits_{R(X)\in\mbox{Ob}(G)}\mbox{id}_{R(X)} = \bigcup\nolimits_{{Y\in\mbox{Ob}(H)}}\mbox{id}_{Y} \qquad \mbox{surjective on objects}
\end{align}
where we have used the fact that $R$ is surjective on objects, which implies that every object of $H$ is in the image of $R$ and that $|\mbox{Ob}(G)|\ge|\mbox{Ob}(H)|$.

The second monoid homomorphism condition~\eqref{eq:monone} is to preserve multiplication, i.e. that for subsets $K$ and $J$ of $G$ we have
\begin{align}
R(K+_GJ)=R(K)+_HR(J).
\end{align}

\noindent Here we recall that for two sets $A$ and $B$, $A+B=\{a + b | a \in A\mbox{ and } b \in B\}$. Thus,
\begin{align}
R(K+_GJ) &= R(\bigcup\nolimits_{k\in K,j\in J}k+_Gj) = \bigcup\nolimits_{k\in K,j\in J}R(k+_Gj) \\
&= \bigcup\nolimits_{k\in K,j\in J}R(k)+_HR(j) \qquad \mbox{def. of group hom. rel.}\\
&= R(K)+_HR(J).
\end{align}
This completes the proof.
\end{proof}

\section{Appendix: Classical Relations}
\label{app:clRel}
%\vspace{-5pt}

In this appendix we list examples of classical relations as calculated by a Mathematica package available at:
\url{https://github.com/willzeng/GroupoidHomRelations}
\begin{align*}
%\vspace{-5pt}
\begin{tabular}{lcl}
Classical relations $\mathbb{Z}_3\to\mathbb{Z}_3$:  & \hspace{50pt} & Classical relations $\mathbb{Z}_4\to\mathbb{Z}_4$:  \\[-10pt]
\begin{tabular}{l}
\{(0,0),   (0,1),   (0,2)\} \\
\{(0,0),(1,1),(2,2)\} \\
\{(0,0),(1,2),(2,1)\}
\end{tabular} & &
\begin{tabular}{l}
\\
\{(0,0),(0,1),(0,2),(0,3)\} \\
\{(0,0),(1,1),(2,2),(3,3)\} \\
\{(0,0),(2,1),(0,2),(2,3)\} \\
\{(0,0),(3,1),(2,2),(1,3)\}
%\end{align*}
\end{tabular}
\end{tabular}
\end{align*}
%\vspace{-50pt}

\noindent The classical relations from  $\mathbb{Z}_2\oplus\mathbb{Z}_2\to\mathbb{Z}_2\oplus\mathbb{Z}_2$ are:
\begin{align*}
\begin{tabular}{cc}
\{(0,2),(2,2),(1,3),(3,3)\} &\qquad \{(0,0),(1,1),(2,2),(3,3)\} \\
\{(0,2),(2,2),(1,3),(2,3)\} &\qquad \{(0,0),(1,1),(2,2),(2,3)\} \\
\{(0,2),(2,2),(0,3),(3,3)\} &\qquad \{(0,0),(0,1),(2,2),(3,3)\} \\
\{(0,2),(2,2),(0,3),(2,3)\} &\qquad \{(0,0),(0,1),(2,2),(2,3)\} \\
\{(2,0),(3,1),(0,2),(1,3)\} &\qquad \{(0,0),(2,0),(1,1),(3,1)\} \\
\{(2,0),(3,1),(0,2),(0,3)\} &\qquad \{(0,0),(2,0),(1,1),(2,1)\} \\
\{(2,0),(2,1),(0,2),(1,3)\} &\qquad \{(0,0),(2,0),(0,1),(3,1)\} \\
\{(2,0),(2,1),(0,2),(0,3)\} &\qquad \{(0,0),(2,0),(0,1),(2,1)\} \\
\end{tabular}
\end{align*}

\end{appendix}

%+Bibliography%
\bibliographystyle{splncs03}
\bibliography{QAlgRel}
%-Bibliography

\end{document}